\documentclass[letterpaper, 10 pt, conference]{ieeeconf}

\usepackage{graphicx}
\usepackage{epsfig}
\usepackage{mathptmx}
\usepackage{times}
\usepackage{amsmath}
\usepackage{amssymb}
\usepackage[section]{placeins}
\usepackage{placeins}
\usepackage{amsmath}
\usepackage{multirow}
\usepackage{graphicx}
\usepackage[table,xcdraw]{xcolor}
\usepackage{epstopdf}
\usepackage[normalem]{ulem}
\usepackage{hyperref}
\usepackage{algorithm,tabularx}
\usepackage{algpseudocode}
\usepackage{amsfonts}
\usepackage{cases}
\usepackage{romannum}
\usepackage{varwidth}
\usepackage{algpseudocode}
\usepackage{caption}
\usepackage{subcaption}
\usepackage{dblfloatfix}
\usepackage{textcomp}
\usepackage{color}
\usepackage{cite}
\usepackage{seqsplit}%
\providecommand{\U}[1]{\protect\rule{.1in}{.1in}}
\providecommand{\U}[1]{\protect\rule{.1in}{.1in}}
\newcommand{\R}{\mathbb{R}}
\newcommand{\N}{\mathbb{N}}

\newcommand{\tsup}[1]{\textsuperscript{#1}}
\newcommand{\mb}[1]{\mathbf{#1}}

\newtheorem{theorem}{Theorem}

\newtheorem{lemma}{Lemma}
\newtheorem{proposition}{Proposition}
\newtheorem{remark}{Remark}
\newtheorem{definition}{Definition}
\useunder{\uline}{\ul}{}
\makeatletter
\newcommand{\multiline}[1]{  \begin{tabularx}{\dimexpr\linewidth-\ALG@thistlm}[t]{@{}X@{}}
#1
\end{tabularx}
}
\makeatother
\algdef{SE}[DOWHILE]{Do}{doWhile}{\algorithmicdo}[1]{\algorithmicwhile\ #1}


\usepackage{picins}
\renewcommand{\theparagraph}{\alph{paragraph})}
\usepackage[T1]{fontenc}
\usepackage{titlesec}
\titleformat{\paragraph}[runin]
{\bfseries\itshape}{\theparagraph}{0.5em}{}[:]

\usepackage{pifont}
\newcommand{\cmark}{\ding{51}}%
\newcommand{\xmark}{\ding{55}}%

\title{\LARGE \bf
Performance-Guaranteed Solutions for Multi-Agent Optimal Coverage Problems using Submodularity, Curvature, and Greedy Algorithms 
}

\author{Shirantha Welikala and Christos G. Cassandras 
\thanks{This work was supported in part by ... }
\thanks{Shirantha Welikala is with the Department of Electrical and Computer Engineering, Stevens Institute of Technology, Hoboken, NJ 07030, USA (Email: \texttt{{\small swelikal@stevens.edu}}). Christos G. Cassandras is with the Division of Systems Engineering and Center for Information and Systems Engineering, Boston University, Brookline, MA 02446, USA (Email: \texttt{{\small cgc@bu.edu}}).}}

\begin{document}

\maketitle

\thispagestyle{plain}
\pagestyle{plain}
\pagenumbering{arabic}

\begin{abstract}
We consider a class of multi-agent optimal coverage problems where the goal is to determine the optimal placement of a group of agents in a given mission space such that they maximize a joint ``coverage'' objective. This class of problems is extremely challenging due to the non-convex nature of the mission space and of the coverage objective. With this motivation, we propose to use a \emph{greedy algorithm} as a means of getting feasible coverage solutions efficiently. Even though such greedy solutions are suboptimal, the \emph{submodularity} (diminishing returns) property of the coverage objective can be exploited to provide performance bound guarantees - not only for the greedy solutions but also for any subsequently improved solutions. Moreover, we show that improved performance bound guarantees (beyond the standard (1-1/e) performance bound) can be established using various \emph{curvature measures} that further characterize the considered coverage problem. In particular, we provide a brief review of all existing popular curvature measures found in the submodular maximization literature, including a recent curvature measure that we proposed, and discuss in detail their applicability, practicality, and effectiveness in the context of optimal coverage problems. Moreover, we characterize the dependence of the effectiveness of different curvature measures (in providing improved performance bound guarantees) on the agent sensing capabilities. Finally, we provide several numerical results to support our findings and propose several potential future research directions.

\end{abstract}

\section{Introduction}

Our research focuses on multi-agent optimal coverage problems, which often arise in critical applications such as (but not limited to) surveillance, security, agriculture, and search and rescue \cite{Rezazadeh2019,Cohen2008}. In these problems, the overall goal is to find an effective placement (decision variable) for the agent team so that they can jointly and optimally ``cover'' (i.e., detect events of interest randomly occurring in) the mission space \cite{Zhong2011}. 

Due to their wide applicability, several variants of multi-agent optimal coverage problems have been extensively studied in the literature \cite{Zhong2011,Luo2019,Welikala2019J1,Yu2022}. Typically, these are formulated as continuous optimization problems inspired by real-world conditions (e.g., continuous mission/decision spaces). However, their corresponding solutions are computationally expensive - unless significant simplifying assumptions are made regarding the particular coverage problem setup. This is mainly due to the overall challenging nature of coverage problems resulting from the often non-linear, non-convex, and non-smooth coverage objectives and non-convex mission spaces involved.   

In this paper, we adopt an alternative approach that has been taken in the literature \cite{Sun2019,Sun2020},  
and formulate the multi-agent optimal coverage problem as a combinatorial optimization problem by discretizing the associated mission/decision space. The coverage objective function, in this setting, is proven to be a \emph{submodular} set function. In other words, the coverage objective function shows \emph{diminishing returns} when the deployed set of agents is expanded. While this combinatorial formulation simplifies the coverage problem to a certain level, the resulting problem, which now takes the form of a submodular maximization problem, is well known to be NP-hard \cite{Corneuejols1977,Nemhauser1978}.    

\emph{Greedy algorithms} are commonly used to solve submodular maximization problems due to their simplicity and computational efficiency. Most importantly, the resulting greedy solutions, even though suboptimal, entertain performance bounds that characterize their proximity to the global optimal solution. The seminal work in \cite{Nemhauser1978} has established a $1-(1-\frac{1}{N})^N$ performance bound, which becomes $(1-\frac{1}{e}) \simeq 0.6321$ as the solution set size (i.e., in the coverage problem, the number of agents) $N\rightarrow \infty$. This implies that the greedy solution is not worse than $63.21\%$ of the global optimal solution. Recent literature on submodular maximization problems has focused on developing improved performance bounds beyond this fundamental performance bound.

To this end, various \emph{curvature measures} have been proposed to further characterize any given submodular maximization problem \cite{Conforti1984,Wang2016,Liu2018,WelikalaJ02021}. These curvature measures provide corresponding performance bounds, which may or may not significantly improve upon the fundamental performance bound - depending on the nature of the considered problem/application. 
However, often these curvature measures are computationally expensive to evaluate. 
Moreover, given the variety of curvature measures available and the variations in their effectiveness with respect to problem parameters, selecting a curvature measure that is likely to provide significantly improved performance bounds for a particular application is challenging.    

Our previous work in \cite{Sun2019} considered a widely studied multi-agent optimal coverage problem (e.g., see \cite{Zhong2011,Welikala2019J1}) and showed that the \emph{total curvature} \cite{Conforti1984} and \emph{elemental curvature} \cite{Wang2016} can collectively provide improved performance bounds irrespective of the used agent sensing capabilities (weak or strong). The subsequent work in \cite{Sun2020} considered a slightly different coverage problem (i.e., the optimal agent team selection for coverage problem) and showcased the effectiveness of the \emph{greedy curvature} \cite{Conforti1984} and \emph{partial curvature} \cite{Liu2018} in providing improved performance bounds. Our most recent work in \cite{WelikalaJ02021} considered the general submodular maximization problem and proposed a new curvature measure called the \emph{extended greedy curvature}. Then, using the coverage problem, the effectiveness of this new curvature measure compared to all other curvature measures mentioned above was illustrated. 
In this paper, we consider a more general and comprehensive coverage problem and investigate the effectiveness of all these curvature measures through theoretical analysis and numerical experiments.

Our contributions are as follows: 
(1) We consider a significantly more general coverage problem (compared to those in \cite{WelikalaJ02021,Sun2019});
(2) Submodularity and several other key properties of the considered set coverage function are established; 
(3) We review five curvature measures that are applicable to the considered coverage problem (to the best of our knowledge, this review is exhaustive);
(4) Special properties and techniques for numerical evaluation of curvature measures in the context of considered coverage problems are discussed; 
(5) We detail the effectiveness of different curvature metrics with respect to the coverage problem parameters (e.g., agent sensing capabilities);
(6) We implement the proposed coverage problem setup in a simulation environment and evaluate different curvature measures, along with their performance bounds, under a diverse set of problem conditions;

\paragraph*{Organization}\ 
The paper is organized as follows. We introduce the considered class of multi-agent optimal coverage problem in Section \ref{Sec:CoverageProblem}. Some notations, preliminary concepts, and the proposed greedy solution are reported in Section \ref{Sec:Preliminaries}. Different curvature measures found in the literature, along with discussions on their applicability, practicality, and effectiveness in the context of optimal coverage problems, are provided in Section \ref{Sec:CurvatureMeasures}. 
Interesting observations, advantages, limitations, and potential future research directions are summarized in Section \ref{Sec:Discussion}. 
Several numerical results obtained from different multi-agent optimal coverage problems \cite{Welikala2019J1} are reported in Section \ref{Sec:CaseStudies} before concluding the paper in Section \ref{Sec:Conclusion}. 

\paragraph*{Notation}\ 
The sets of real and natural numbers are denoted by $\R$ and $\N$, respectively. $\R_{\geq 0}$ represents the set of non-negative real numbers, $\R^n$ denotes the set of $n$-dimensional real (column) vectors, $\mathbb{N}_n \triangleq \{1,2,\ldots,n\}$, $\N_n^0 \triangleq \N_n \cup \{0\}$, and $[a,b]\triangleq \{x: x\in\R, a\leq x \leq b\}$. $\Vert \cdot \Vert$ represents the Euclidean norm, $\vert \cdot \vert$ denotes the scalar absolute value or set cardinality (based on the type of the argument), $\left \lfloor{\cdot}\right \rfloor$ denotes the floor operator, and $\mathbf{1}_{\{\cdot\}}$ denotes the indicator function. Given two sets $A$ and $B$, the set subtraction operator is denoted as $A-B = A \backslash B = A\cap B^c$. $2^X$ denotes the power set of a set $X$ and $\emptyset$ is the empty set.

\section{Multi-Agent Optimal Coverage Problem}
\label{Sec:CoverageProblem}
We begin by providing the details of the considered multi-agent optimal coverage problem. The goal of the considered coverage problem is to determine an optimal placement for a given team of agents (e.g., sensors, cameras, guards, etc.) in a given mission space that maximizes the probability of detecting events that occur randomly over the mission space. 

We model the \emph{mission space} $\Omega$ as a convex polytope in $\R^n$ that may also contain $h$ polytopic (and possibly non-convex) \emph{obstacles} $\{\Psi_i:\Psi_i\subset\Omega, i\in\mathbb{N}_h\}$. The characteristics of the obstacles are such that they: (1) limit the agent placement to the feasible space $\Phi \triangleq \Omega \backslash \cup_{i\in\N_h} \Psi_i$, (2) constrain the sensing capabilities of the agents via obstructing their line of sight, and (3) are in areas where no events of interest occur. 

To model the likelihood of random \emph{events} occurring over the mission space, an \emph{event density function} $R:\Omega \rightarrow \R_{\geq0}$ is used, where $R(x) = 0, \forall x \not\in \Phi$ and $\int_\Omega R(x)dx <\infty$. Note that when no prior information is available on $R(x)$, one can use $R(x)=1, \forall x\in \Phi$. 

To detect these random events, $N$ \emph{agents} are to be placed inside the feasible space $\Phi$. The placement of this team of $N$ agents (i.e., the decision variable) is denoted in a compact form by $s = [s_1, s_2, \ldots, s_N]\in\R^{m\times N}$, where each $s_i, i\in\N_N$ represents an agent placement such that $s_i\in \Phi$. 

The ability of an agent to \emph{detect} events is limited by visibility obstruction from obstacles and the agent's sensing capabilities. For an agent placed at $s_i \in \Phi$, its \emph{visibility region} is defined as 
$$
V(s_i) \triangleq \{x: (qx+(1-q)s_i)\in \Phi, \forall q \in [0,1] \}.
$$ 
In terms of agent \emph{sensing capabilities}, agents are assumed to be homogeneous, where each agent has a finite \emph{sensing radius} $\delta \in \R_{\geq 0}$ and \emph{sensing decay rate} $\lambda\in\R_{\geq 0}$. In particular, the probability of an agent placed at $s_i\in \Phi$ detecting an event occurring at $x\in \Phi$ is described by a \emph{sensing function} defined as 
\begin{equation}
\label{Eq:SensingFunction}
p(x,s_i)\triangleq e^{-\lambda \Vert x - s_i\Vert}\cdot\mathbf{1}_{\{x\in V(s_i)\}}.
\end{equation}
Figure \ref{Fig:Geometry} illustrates visibility regions and the corresponding sensing functions in a 2-D mission space with two agents and three obstacles.

\begin{figure}[!b]
    \centering
    \includegraphics[width=2.5in]{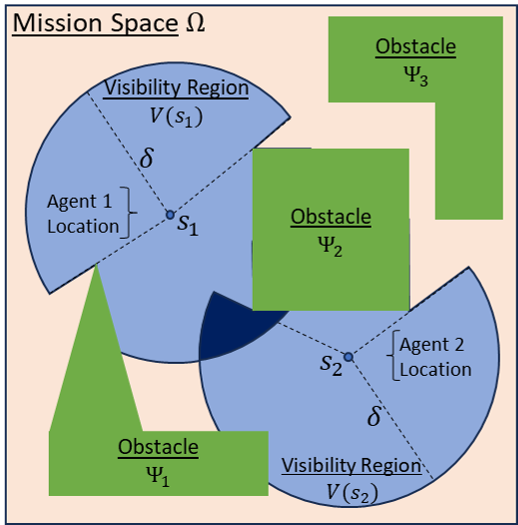}
    \caption{A mission space with two agents.}
    \label{Fig:Geometry}
\end{figure}

Given the placement $s$ of the team of agents, their combined ability to detect an event occurring at $x \in \Phi$ is characterized by a \emph{detection function} $P(x,s)$. One popular detection function is the \emph{joint detection function} given by 
\begin{equation}\label{Eq:JointDetection}
P_J(x,s) \triangleq 1-\prod_{i\in\N_N}(1-p(x,s_i)),
\end{equation}
which represents the probability of detection by at least one agent (assuming agents detect events independently from each other). Another widely used detection function is the \emph{max detection function} given by 
\begin{equation}\label{Eq:MaxDetection}
P_M(x,s) \triangleq \max_{i\in\N_N} p(x,s_i),
\end{equation}
which represents the maximum probability of detection by any agent. The following remark summarizes the pros and cons of using \eqref{Eq:JointDetection} or \eqref{Eq:MaxDetection} as the detection function $P(x,s)$.  

\begin{remark}\label{Rm:DetectionFunction}
The joint detection function \eqref{Eq:JointDetection} aggregates the contributions of all agents and, thus, offers a comprehensive view of coverage. This method is suitable for applications that benefit from the collective capabilities of the agent team. However, it can be computationally demanding to compute. 
On the other hand, the max detection function \eqref{Eq:MaxDetection} prioritizes the most effective agent at each point and, thus, offers a conservative estimate of coverage. This method is simpler and may be preferred in critical applications where ensuring the highest detection probability at every point is paramount. However, it can lead to under-utilization of the agent team as it does not fully account for the combined coverage provided by multiple sensors.
Ultimately, the choice between these detection functions should consider the coverage application's comprehensiveness and reliability requirements, the nature of the agents, and the available computational resources.
\end{remark}

Motivated by the contrasting nature of the joint and max detection functions \eqref{Eq:JointDetection}-\eqref{Eq:MaxDetection}, in this paper, we consider the detection function
\begin{equation}\label{Eq:DetectionFunction}
    P(x,s) \triangleq \theta P_J(x,s) + (1-\theta) P_M(x,s),
\end{equation}
where $\theta\in[0,1]$ is a predefined weight.

Using the notions of event density function $R(x)$ and detection function $P(x,s)$ introduced above, the considered \emph{coverage function} is defined as 
\begin{equation}\label{Eq:CoverageObjective}
H(s) \triangleq \int_\Omega R(x)P(x,s)dx.     
\end{equation}
Consequently, the considered multi-agent optimal coverage problem can be stated as
\begin{equation}\label{Eq:CoverageProblem}
 s^* = \underset{s:s_i\in \Phi, i\in \N_N}{\arg\max}\ H(s).  
\end{equation}

\paragraph*{Continuous Optimization Approach}\ 
The optimal coverage problem \eqref{Eq:CoverageProblem} involves a non-convex feasible space and a non-convex, non-linear, and non-smooth objective function. While the prior is due to the presence of obstacles in the mission space, the latter is due to the nature of the event density, sensing, joint detection, and coverage function forms. Consequently, it is extremely difficult to solve this problem without using: (1) standard global optimization solvers that are computationally expensive, (2) systematic gradient-based solvers that require extensive domain knowledge, or (3) Voronoi partition techniques that require significant limiting assumptions (e.g., convexity \cite{Yu2022} and connectivity \cite{Luo2019}).   

\paragraph*{Combinatorial Optimization Approach}\ Motivated by the aforementioned challenges associated with different continuous optimization approaches, in this paper, we take a combinatorial optimization approach to the formulated the multi-agent optimal coverage problem \eqref{Eq:CoverageProblem}. This requires reformulating \eqref{Eq:CoverageProblem} as a set function maximization problem.

First, we discretize the feasible space $\Phi$ formulating a \emph{ground set} $X = \{x_l:x_l\in \Phi, l\in\N_M\}$. For this discretization, a grid can be placed inside the mission space, and then all the grid points except for those that fall inside obstacles can be considered as the ground set\footnote{If obstacle and/or grid resolution are large (compared to the mission space dimensions), the resulting grid-based ground set needs to be further optimized to ensure uniformity over feasible space. This can be achieved by assuming all points in the obtained ground set are occupied by a set of ``virtual'' agents and then executing a simplified multi-agent optimal coverage solver. The resulting virtual agent locations can then be treated as the ground set.}.

Second, upon formulating the ground set $X$, we define a \emph{set variable} $S = \{s_i:i\in\N\}$ to represent the set of selected locations for agent placement. As we are interested in placing only $N$ agents strictly in locations selected from the ground set $X$, we need to constrain this set variable $S$ such that $S \in \mathcal{I}^N \triangleq \{Y: Y \subseteq X, \vert Y \vert \leq N\}$. It is worth noting that a set system of the form $(X,\mathcal{I}^N)$ is known as a \emph{uniform matroid} of rank $N$, and a set constraint of the form $S\in\mathcal{I}^N$ is known as a uniform matroid constraint of rank $N$. 

To represent the coverage function value of the agent placement defined by the set variable $S$, using the coverage function \eqref{Eq:CoverageObjective}, we next define a \emph{set coverage function} as: 
\begin{equation}\label{Eq:SetCoverageFunction}
    H(S) \triangleq \int_\Omega R(x)P(x,S)dx,
\end{equation}
where $P(x,S)$ represents a \emph{set detection function}. Inspired by \eqref{Eq:JointDetection}-\eqref{Eq:DetectionFunction}, this set detection function $P(x,S)$ is selected as:
\begin{equation}\label{Eq:SetDetectionFunction}
P(x,S) \triangleq \theta P_J(x,S) + (1-\theta)P_M(x,S)
\end{equation}
where 
\begin{equation}\label{Eq:SetJointDetection}
P_J(x,S) \triangleq 1-\prod_{s_i\in S}(1-p(x,s_i)),
\end{equation}
\begin{equation}\label{Eq:SetMaxDetection}
P_M(x,S) \triangleq \max_{s_i\in S} p(x,s_i).    
\end{equation}

Finally, we restate the original multi-agent optimal coverage problem \eqref{Eq:CoverageProblem} as a set function maximization problem:
\begin{equation}\label{Eq:SetCoverageProblem}
S^* = \underset{S\in\mathcal{I}^N}{\arg\max}\ H(S).
\end{equation}

Since the size of the set variable search space of the formulated set function maximization problem \eqref{Eq:SetCoverageProblem} is combinatorial (in particular $\vert \mathcal{I}^N \vert = \sum_{r=0}^N {M \choose r} =  \sum_{r=0}^M \frac{M!}{(M-r)!r!}$), obtaining an optimal solution (i.e., $S^*$) for it is impossible without significant simplifying assumptions. Therefore the overall goal is to obtain a candidate solution for \eqref{Eq:SetCoverageProblem} (say $S^G$) in an efficient manner with some guarantees on its coverage performance $H(S^G)$ with respect to optimal coverage performance $H(S^*)$.

One obvious approach to obtaining such a candidate solution efficiently is via a vanilla \emph{greedy algorithm} as given in Alg. \ref{Alg:GreedyAlgorithm}. Note that it uses the \emph{marginal coverage function} defined as
\begin{equation}\label{Eq:MarginalCoverage}
    \Delta H(y \vert S^{i-1}) \triangleq H(S^{i-1}\cup\{y\}) - H(S^{i-1})
\end{equation}
to iteratively determine the optimal agent placements until $N$ such agent placements have been chosen. Note that, in \eqref{Eq:MarginalCoverage}, we require $S^{i-1}\in\mathcal{I}^N$ and $S^{i-1}\cup\{y\} \in\mathcal{I}^N$, and clearly, $y \in X \backslash S^{i-1}$ (to ensure $\Delta H(y \vert S^{i-1})>0$). While the latter condition means that no two agents are allowed at the same place in the mission space, if necessary, by appropriately defining the ground set $X$, we can relax such placement constraints. 

Let us define the notion of \emph{marginal detection function} as
\begin{equation}\label{Eq:MarginalDetection}
    \Delta P(x, y \vert S^{i-1}) = P(x, S^{i-1}\cup\{y\}) - P(x, S^{i-1}).
\end{equation}
This definition is motivated by the linear relationship \eqref{Eq:SetCoverageFunction} between the set coverage function $H(S)$ and the set detection function $P(x,S)$ with respect to the set argument $S$, because a similar linear relationship exists between the corresponding marginal functions $\Delta H(y \vert S)$ \eqref{Eq:MarginalCoverage} and $\Delta P(x,y\vert S)$ \eqref{Eq:MarginalDetection}. In the sequel, we exploit these linear relationships to conclude certain set function properties of $H(S)$ using those of $P(x,S)$ and $\Delta P(x,y\vert S)$.

Finally, we point out that, the notation $S^i \triangleq \{s^1,s^2,\ldots,s^i\}$ (with $S^0=\emptyset$ and $S^N=S^G$) used to represent the greedy solution obtained after $i$ greedy iterations in Alg. \ref{Alg:GreedyAlgorithm} will be used more liberally for any $i\in\{0,1,2,\ldots,M\}$ in the sequel.

\begin{algorithm}[!h]
\caption{The greedy algorithm to solve \eqref{Eq:SetCoverageProblem}}\label{Alg:GreedyAlgorithm}
\begin{algorithmic}[1]
\State $i=0$; $S^i = \emptyset$;  \Comment{Greedy iteration index and solution}
\For{$i=1,2,3,\ldots,N$}
    \State $s^{i} = \arg \max_{\{y:S^{i-1} \cup \{y\} \in \mathcal{I}^N\}} \Delta H(y \vert S^{i-1})$; \Comment{New item}
    \State $S^{i} = S^{i-1} \cup \{s^{i}\}$; \Comment{Append the new item}
\EndFor
\State $S^G := S^N$;  \textbf{Return} $S^G$;
\end{algorithmic}
\end{algorithm}

\section{The Greedy Solution with Performance Bound Guarantees}
\label{Sec:Preliminaries}


In this section, we will show that the greedy solution $S^G$ obtained using Alg. \ref{Alg:GreedyAlgorithm} for the optimal coverage problem \eqref{Eq:SetCoverageProblem} is not only computationally efficient but also entertains performance guarantees with respect to the global optimal coverage performance $H(S^*)$. For this, we first need to introduce some standard set function properties.

For generality, let us consider an arbitrary set function $F:2^Y \rightarrow \R$ defined over a finite ground set $Y$. Similar to before, the corresponding marginal function is defined as 
\begin{equation}
    \Delta F(y \vert A) \triangleq F(A\cup\{y\}) - F(A)
\end{equation}
to represent the gain of set function value due to the addition of an extra element $y \in Y \backslash A$ to the set $A$. Note that we can use this marginal function notation more liberally as $\Delta F(B \vert A) \triangleq f(A\cup B) - f(A)$ for any $A, B \subseteq Y$ (even allowing $B \not\subseteq Y\backslash A$).

\begin{definition} \cite{WelikalaJ02021} \label{Def:Submodularity} The set function $F:2^Y \rightarrow \R$ is:
\begin{enumerate}
    \item \emph{normalized} if $F(\emptyset) = 0$;
    \item \emph{monotone} if $\Delta F(y \vert A)\geq 0$ for all $y,A$ where $A \subset Y$ and $y\in Y\backslash A$, or equivalently, if $F(B) \leq F(A)$ for all $B,A$ where $B \subseteq A \subseteq Y$;
    \item \emph{submodular} if $\Delta f(y\vert A) \leq \Delta f(y\vert B)$ for all $y,A,B$ where $B\subseteq A \subset Y$ and $y\in Y \backslash A$, or equivalently, \\if $F(A\cup B) + F(A\cap B) \leq F(A) + F(B)$ for all $A,B\subseteq Y$,
    \item a \emph{polymatroid} set function if it is normalized, monotone and submodular \cite{Liu2018,Boros2003}. 
\end{enumerate}
\end{definition}

It is worth noting that the first condition outlined for the submodularity property in Def. \ref{Def:Submodularity}-(3) is more commonly known as the \emph{diminishing returns} property. 

The following lemma and the theorem establish the polymatroid nature of the set coverage function $H(S)$ \eqref{Eq:SetCoverageFunction}.

\begin{lemma}\label{Lm:LinearityOfSubmodularity}
With respect to a common ground set, any positive linear combination of arbitrary polymatroid set functions is also a polymatroid set function.    
\end{lemma}
\begin{proof}
Consider $n$ polymatroid set functions $F_1,F_2,\ldots,F_n$, where each is defined over a common ground set $Y$. Let $F\triangleq \sum_{i\in\N_n} \alpha_i F_i$, where $\alpha_i \geq 0, \forall i \in\N_n$. Clearly, $F(\emptyset)=0$ as $F_i(\emptyset)=0$ and $\alpha_i \geq 0, \forall i\in\N_n$. Therefore, $F$ is normalized. With respect to any $A,B$ such that $B \subseteq A \subseteq Y$, note that $F(B)\geq F(A)$ holds as $F_i(B) \leq F_i(A)$ and $\alpha_i \geq 0, \forall i\in\N_n$. Thus, $F$ is monotone. Using the same arguments, the submodularity and, thus, the polymatroid nature of $F$ can be established.
\end{proof}

\begin{theorem}\label{Th:SetCoverageSubmodularity}
The set coverage function $H(S)$ \eqref{Eq:SetCoverageFunction} is a polymatroid set function.
\end{theorem}
\begin{proof}
From \eqref{Eq:SetCoverageFunction} and \eqref{Eq:SetDetectionFunction}, it is clear that the set coverage function $H(S)$ can be viewed as a linear combination of set detection function components $P_J(x,S)$ \eqref{Eq:SetJointDetection} and $P_M(x,S)$ \eqref{Eq:SetMaxDetection} - with respect to the common ground set $X$. Therefore, in light of Lm. \ref{Lm:LinearityOfSubmodularity}, to prove $H(S)$ a polymatroid set function, we only have to show $P_J(x,S)$ and $P_M(x,S)$ are polymatroid set functions with respect to any feasible space point $x\in\Phi$. 

Let us first consider the set function $P_J(x,S)$ \eqref{Eq:SetJointDetection}. By definition, $P_J(x,\emptyset) = 0$, and thus, $P_J(x,S)$ is normalized. The corresponding marginal function \eqref{Eq:MarginalDetection} can be derived as:
\begin{align}
\Delta P_J(x,y \vert A) =&\ P_J(x,A\cup\{y\}) - P_J(x,A) \nonumber\\    
=&\ -\prod_{s_i\in A\cup\{y\}}(1-p(x,s_i)) +\prod_{s_i\in A}(1-p(x,s_i)) \nonumber\\
=&\ p(x,y)\prod_{s_i\in A}(1-p(x,s_i)), \label{Eq:Th:SetCoverageSubmodularityStep1}
\end{align}
for any $A \subset X$ and $y\in X\backslash A$. 
From \eqref{Eq:Th:SetCoverageSubmodularityStep1}, it is clear that $\Delta P_J(x,y \vert A) \geq 0$ for all $A \subset X$ and $y\in X\backslash A$. Therefore, $P_J(x,S)$ is monotone. Note also that the product term in \eqref{Eq:Th:SetCoverageSubmodularityStep1} diminishes with the growth of the set $A$. This property can be used to conclude that $\Delta P_J(x,y \vert A) \leq \Delta P_J(x,y \vert B)$, for all $y,A,B$ where $B\subseteq A \subset X$ and $y\in Y\backslash A$. Hence $P_J(x,S)$ is submodular. Therefore, $P_J(x,S)$ is a polymatroid set function.

Finally, let us consider the set function $P_M(x,S)$ \eqref{Eq:SetMaxDetection}. Again, by definition, $P_M(x,\emptyset) = 0$ implying that $P_M(x,S)$ is normalized. The corresponding marginal function \eqref{Eq:MarginalDetection} can be derived as:
\begin{align}
\Delta P_M(x,y \vert A) =&\ P_M(x,A\cup\{y\}) - P_M(x,A) \nonumber\\    
=&\ \max_{s_i\in A\cup\{y\}} p(x,s_i) - \max_{s_i\in A} p(x,s_i)\nonumber\\
=&\ \max\{p(x,y)-\max_{s_i\in A} p(x,s_i),0\}, \label{Eq:Th:SetCoverageSubmodularityStep2}
\end{align}
for any $A \subset X$ and $y\in X\backslash A$. From \eqref{Eq:Th:SetCoverageSubmodularityStep2}, it is clear that $\Delta P_M(x,y \vert A) \geq 0$ for all $A \subset X$ and $y\in X\backslash A$. Therefore, $P_M(x,S)$ is monotone. Note also that the first term inside the outer max operator in \eqref{Eq:Th:SetCoverageSubmodularityStep2} diminishes with the growth of the set $A$. This property implies that $\Delta P_M(x,y \vert A) \leq \Delta P_M(x,y \vert B)$, for all $y,A,B$ where $B \subseteq A \subset X$ and $y \in Y\backslash A$. Hence $P_M(x,S)$ is submodular. Consequently, $P_M(x,S)$ is a polymatroid set function. This completes the proof.
\end{proof}

As a direct result of this polymatroid nature of the set coverage function $H(S)$ \eqref{Eq:SetCoverageFunction}, we can characterize the proximity of the performance of the greedy solution (i.e., $H(S^G)$) to that of the globally optimal solution (i.e., $H(S^*)$). For this characterization, we particularly use the notion of a \emph{performance bound}, (denoted by $\beta$) defined as a theoretically established lower bound for the ratio $\frac{H(S^G)}{H(S^*)}$, i.e., 
\begin{equation}\label{Eq:Def:PerformanceBound}
    \beta \leq \frac{H(S^G)}{H(S^*)}.
\end{equation}
Having a performance bound $\beta$ close to $1$ implies that the performance of the greedy solution is close to that of the global optimal solution. Consequently, $\beta$ can also serve as an indicator of the effectiveness of the greedy approach to solve the interested optimal coverage problem \eqref{Eq:SetCoverageProblem}.   

The seminal work \cite{Nemhauser1978} has established a performance bound (henceforth called the \emph{fundamental performance bound}, and denoted by $\beta_f$) for polymatroid set function maximization problems, which, in light of Th. \ref{Th:SetCoverageSubmodularity}, is applicable to the optimal coverage problem \eqref{Eq:SetCoverageProblem} as: 
\begin{equation}\label{Eq:FundamentalPerformanceBound}
    \beta_f \triangleq 1-\left(1-\frac{1}{N}\right)^N \leq \frac{H(S^G)}{H(S^*)}.
\end{equation}
Note that, while $\beta_f$ decreases with the number of agents $N$, it is lower-bounded by $1-\frac{1}{e} \simeq 0.6321$, because $\lim_{N\rightarrow\infty} \beta_f = (1-\frac{1}{e})$. This implies that the coverage performance of the greedy solution will always be not worse than 63.21\% of the maximum achievable coverage performance.

Moreover, as shown in \cite{WelikalaJ02021}, upon obtaining the greedy solution $S^G$ from Alg. \ref{Alg:GreedyAlgorithm}, any subsequently improved solution $\bar{S}^G\in\mathcal{I}^N$ (e.g., via a gradient process \cite{Sun2019,Sun2020} or an interchange scheme \cite{Welikala2019P3}), will have an improved performance bound $\bar{\beta}$ than the original performance bound $\beta$ such that  
$$\beta \leq \bar{\beta} \triangleq \beta * \frac{H(\bar{S}^G)}{H(S^G)} \leq \frac{H(\bar{S}^G)}{H(S^*)}.$$ 

Besides (and independently from) improving the original greedy solution, as we will see in the next section, an improved performance bound can also be achieved by exploiting certain characteristics (called \emph{curvature measures}) of the interested set function maximization problem. 

Before moving on, we provide a minor technical lemma along with a theorem that establishes the polymatroid nature of the marginal coverage function $\Delta H (B \vert A)$ with respect to both of its set arguments $A$ and $B$. 

\begin{lemma}\label{Lm:MaxMinusMax}
For any $a,b \in \R$, $\max \{a,b\} - \max\{c,d\} = \max\{\min\{a-c,a-d\},\min\{b-c,b-d\}\}$.
\end{lemma}
\begin{proof}
The proof follows from simplifying the left-hand side (LHS) of the given relationship:
\begin{align*}
\mbox{LHS} =&\ \max\{a-\max\{c,d\},b-\max\{c,d\}\}\\
=&\ \max\{a+\min\{-c,-d\},b+\min\{-c,-d\}\}\\
=&\ \max\{\min\{a-c,a-d\},\min\{b-c,b-d\}\}.
\end{align*}
\end{proof}

\begin{theorem}\label{Th:SubmodularityOfMarginalCoverage}
For a fixed set $A \subset X$, the marginal coverage function $G_A(B) \triangleq \Delta H (B \vert A)$ is a polymatroid set function over $B \subseteq X \backslash A$. Moreover, for a fixed set $B \subset X$, the affine negated marginal coverage function $G_B(A) \triangleq -\Delta H (B \vert A)+H(B)$ is a polymatroid set function over $A \subseteq X \backslash B$.
\end{theorem} 
\begin{proof}
The first result directly follows from the polymatroid nature of the set coverage function $H(S)$  \eqref{Eq:SetCoverageFunction} established in Th. \ref{Th:SetCoverageSubmodularity} and the theoretical result in \cite[Lm. 1]{WelikalaJ02021}.

To establish the second result, first, note that $G_B(\emptyset) = -\Delta H (B \vert \emptyset)+H(B) = -H(B)+H(\emptyset)+H(B) = 0$. Therefore the set function $G_B(\cdot)$ is normalized. 

Second, consider a set $C \subseteq A \subseteq X\backslash B$. Then, 
$G_B(A)-G_B(C) = \Delta H (B \vert C) - \Delta H (B \vert A)$. To deduce the sign of $\Delta H (B \vert C) - \Delta H (B \vert A)$, we need to inquire about the sign of $\Delta P(x,B\vert C) - \Delta P(x, B \vert A)$ for all $x \in \Omega$. Note that
\begin{align*}
\Delta P_J(x,B\vert C) 
=&\ P_J(x,B\cup C) - P_J(x,C) \\
=&\ \prod_{s_i\in C}(1-p(x,s_i)) - \prod_{s_i\in B\cup C} (1-p(x,s_i))\\
=&\ (1-\prod_{s_i\in B}(1-p(x,s_i)))\prod_{s_i\in C} (1-p(x,s_i))
\end{align*}
where the last step is due to $C\subseteq X \backslash B$. Similarly, we get 
$$
\Delta P_J(x,B\vert A) = (1-\prod_{s_i\in B}(1-p(x,s_i)))\prod_{s_i\in A} (1-p(x,s_i)).
$$
Since $C\subseteq A$, from the above two results, it is clear that $\Delta P_J(x,B\vert C) - \Delta P_J(x, B \vert A) \geq 0$ for all $x \in \Omega$. 
Note also that,
\begin{align*}
    \Delta P_M(x,B \vert C) =&\ P_M(x,B\cup C) - P_M(x,C)\\
    =&\ \max_{s_i \in B\cup C} p(x,s_i)  - \max_{s_i\in C} p(x,s_i)\\
    =&\ \max\{\max_{s_i \in B} p(x,s_i) - \max_{s_i \in C} p(x,s_i),0\}
\end{align*}
where the last step is due to $C\subseteq X \backslash B$. Similarly, we get 
$$
\Delta P_M(x,B \vert A) = 
\max\{\max_{s_i \in B} p(x,s_i) - \max_{s_i \in A} p(x,s_i),0\}.
$$
Since $C\subseteq A$, from the above two results, we get $P_M(x,B\vert C) - \Delta P_M(x, B \vert A) \geq 0$ for all $x \in \Omega$. Using the above two main conclusions with \eqref{Eq:SetDetectionFunction}, we get $\Delta P(x,B\vert C) - \Delta P(x, B \vert A)\geq$ for all $x \in \Omega$. This result, together with \eqref{Eq:SetCoverageFunction}, implies that $\Delta H (B \vert C) - \Delta H (B \vert A) \geq 0$. Therefore, the set function $G_B(\cdot)$ is monotone.

Finally, let us consider an element $y\in (X\backslash B)\backslash A$ and a set $C \subseteq A \subseteq X\backslash B$. For submodularity of $G_B(\cdot)$, we require 
\begin{align*}
&G_B(A\cup\{y\})-G_B(A) \leq G_B(C\cup\{y\})-G_B(C)\\
&\iff G_B(A\cup\{y\})-G_B(C\cup\{y\}) \leq G_B(A)-G_B(C)\\
&\iff \Delta H (B \vert C\cup\{y\}) - \Delta H (B \vert A\cup\{y\}) \\
&\quad\quad\quad\quad \quad\quad\quad\quad \quad\quad\quad\quad
\leq \Delta H (B \vert C) - \Delta H (B \vert A)\\
&\iff \Delta P(x,B\vert C\cup\{y\}) - \Delta P(x, B \vert A \cup \{y\}) \\
&\quad\quad\quad\quad \quad\quad\quad\quad 
\leq \Delta P(x,B\vert C) - \Delta P(x, B \vert A), \forall x\in\Omega.
\end{align*}
Note that, using the previously obtained $P_J(x,B\vert C)$ and $\Delta P_J(x,B\vert A)$ expressions, we get 
\begin{align*}
\Delta P_J(x,B\vert C)-\Delta P_J(x,B\vert A) = 
(1-\prod_{s_i\in B}(1-p(x,s_i)))\\
\times (1-\prod_{s_i\in A \backslash C} (1-p(x,s_i)))\prod_{s_i\in C} (1-p(x,s_i)).
\end{align*}
In this expression, if $A\rightarrow A\cup\{y\}$ and $C \rightarrow C \cup\{y\}$, only the last product term changes; in particular, it decreases, and so does the overall expression.
Note also that, using the previously obtained $P_M(x,B\vert C)$ and $\Delta P_M(x,B\vert A)$ terms and the notation $P_S \triangleq \max_{s_i\in S} p(x,s_i)$, we get (also using Lm. \ref{Lm:MaxMinusMax})
\begin{align*}
\Delta P_M(x,B\vert C) &- \Delta P_M(x,B\vert A) \\
=&\ \max\{P_B-P_C,0\}-\max\{P_B-P_A,0\}\\
=&\  \max\{\min\{P_A-P_C,P_B-P_C\},\min\{P_A-P_B,0\}\}\\
=&\ \max\{\min\{P_A,P_B\}-P_C,\min\{P_A,P_B\}-P_B\}\\
=&\ \min\{P_A,P_B\} - \min\{P_C,P_B\}.
\end{align*}
In this expression, if $A\rightarrow A\cup\{y\}$ and $C \rightarrow C \cup\{y\}$, the increment in $P_C$ will be larger than the increment in $P_A$ as $C \subseteq A$, and thus the overall expression decreases. 
Using these two conclusions with \eqref{Eq:SetDetectionFunction}, we get 
$\Delta P(x,B\vert C\cup\{y\}) - \Delta P(x, B \vert A \cup \{y\}) \leq \Delta P(x,B\vert C) - \Delta P(x, B \vert A), \forall x\in\Omega$, which implies that $G_B(\cdot)$ is submodular. 

Consequently, $G_B(\cdot)$ is a polymatroid set function. This completes the proof.
\end{proof}

The above result further emphasizes the deep polymatroid characteristics of optimal coverage problems. Moreover, as we will see in the sequel, it enables the computation of efficient estimates for some curvature measures introduced in the next section.

\section{Improved Performance Bound Guarantees using Curvature Measures}
\label{Sec:CurvatureMeasures}

In this section, we will discuss several improved performance-bound guarantees that are applicable to the considered optimal coverage problem \eqref{Eq:SetCoverageProblem} and its greedy solution $S^G$ given by Alg. \ref{Alg:GreedyAlgorithm}. The goal is to obtain tighter performance bounds for $S^G$, i.e., closer to 1 compared to $\beta_f$ in \eqref{Eq:FundamentalPerformanceBound}. This is important as such a performance bound will accurately characterize the proximity of $S^G$ to $S^*$, and thus allow making informed decisions regarding spending extra resources (e.g., computational power, agents and sensing capabilities) to seek a further improved coverage solution beyond $S^G$.

As mentioned earlier, curvature measures are used to obtain such improved performance bounds. These curvature measures are dependent purely on the underlying objective function, the ground set, and the feasible space, which, in the considered optimal coverage problem, are $H(S)$, $X$, and $\mathcal{I}^N$, respectively. In this section, we will review five established curvature measures and their respective performance bounds, outlining their unique characteristics, strengths, and weaknesses in their application to the considered optimal coverage problem \eqref{Eq:SetCoverageProblem}.

\subsection{Total Curvature \cite{Conforti1984}}
\label{SubSec:TotalCurvature}

By definition, the \emph{total curvature} of \eqref{Eq:SetCoverageProblem} is given by  
\begin{equation}\label{Eq:TotalCurvatureCoefficientTheory}
    \alpha_t \triangleq \max_{\substack{y \in X} }\left[1 - \frac{\Delta H(y \vert X \backslash \{y\})}{\Delta H(y \vert \emptyset)}\right].
\end{equation}
The corresponding performance bound $\beta_t$ is given by
\begin{equation}\label{Eq:TotalCurvatureBoundTheory}
    \beta_t \triangleq \frac{1}{\alpha_t} \left[ 1 - \left( 1 -\frac{\alpha_t}{N} \right)^N \right] \leq \frac{H(S^G)}{H(S^*)}.
\end{equation}

From comparing \eqref{Eq:TotalCurvatureBoundTheory} and \eqref{Eq:FundamentalPerformanceBound}, it is clear that when $\alpha_t\rightarrow 1$, the corresponding performance bound $\beta_t\rightarrow\beta_f$ (i.e., no improvement). However, on the other hand, when $\alpha_t\rightarrow 0$, the corresponding performance bound $\beta_t\rightarrow 1$ (i.e., a significant improvement). Moreover, it can be shown that $\beta_t$ is monotonically decreasing in $\alpha_t$. Using the above three facts and \eqref{Eq:TotalCurvatureCoefficientTheory}, it is easy to see that the improvement in the performance bound is proportional to the magnitude of:
\begin{equation}\label{Eq:TotolCurvatureGamma}
    \gamma_t \triangleq \min_{y \in X}\left[\frac{\Delta H(y \vert X \backslash \{y\})}{\Delta H(y \vert \emptyset)}\right].
\end{equation}
The diminishing returns (submodularity) property of $H$ implies $\frac{\Delta H(y \vert X \backslash \{y\})}{\Delta H(y \vert \emptyset)} \leq 1, \forall y\in X$. Therefore, $\gamma_t$ is large only when $\frac{\Delta H(y \vert X \backslash \{y\})}{\Delta H(y \vert \emptyset)} \simeq 1, \forall y\in X$. In other words, a significantly improved performance bound from the total curvature measure can only be obtained when $H$ is just ``weakly'' submodular (i.e., when $H$ is closer to being modular rather than submodular). This is also clear from simplifying the condition $\frac{\Delta H(y \vert X \backslash \{y\})}{\Delta H(y \vert \emptyset)} \simeq 1, \forall y \in X$ using \eqref{Eq:MarginalCoverage}, as it leads to
\begin{equation}
\label{Eq:ConditionForWeakSubmodularity}
     H(X) \simeq H(y) + H(X\backslash \{y\}), \quad \mbox{for all } y\in X
\end{equation}
which holds whenever $H$ is modular. 

In particular, as $H$ is the set coverage function \eqref{Eq:SetCoverageFunction}, the above condition \eqref{Eq:ConditionForWeakSubmodularity} holds (leading to improved performance bounds) when an agent deployed at any $y\in X$ and all other agents deployed at $X \backslash \{y\}$ contribute to the coverage objective independently in a modular fashion. This happens when the ground set $X$ is very sparse and/or when the agents have significantly weak non-overlapping sensing capabilities (i.e., small range $\delta$ and high decay $\lambda$ in \eqref{Eq:SensingFunction}). 

However, the condition \eqref{Eq:ConditionForWeakSubmodularity} is easily violated (leading to poor performance bounds) if 
$$
H(X) \ll H(y) + H(X\backslash \{y\}), \quad \mbox{ for some } y\in X. 
$$
To interpret this condition using \eqref{Eq:SetCoverageFunction}, we need to consider the corresponding detection function \eqref{Eq:SetDetectionFunction} requirement: 
$$
P(x,X) \ll P(x,\{y\}) + P(x,X\backslash \{y\}), \quad \mbox{for some } y\in X
$$
for a majority of $x\in\Phi$. Now, using \eqref{Eq:Th:SetCoverageSubmodularityStep1} and \eqref{Eq:Th:SetCoverageSubmodularityStep2}, we get 
\begin{align*}
0\ll&\ \theta(p(x,y)(1-\prod_{s_i\in X\backslash \{y\}}(1-p(x,s_i)))\\
&+(1-\theta) (p(x,y)-\max\{p(x,y)-\max_{s_i\in X\backslash \{y\}}p(x,s_i),0\}),
\end{align*}
where the second term can be further simplified to obtain:
\begin{align*}
0\ll&\ \theta(p(x,y)(1-\prod_{s_i\in X\backslash \{y\}}(1-p(x,s_i)))\\
&+(1-\theta) (\min\{\max_{s_i\in X\backslash \{y\}}p(x,s_i),p(x,y)\}).  
\end{align*}
Since $\theta \in [0,1]$, we need to consider both terms above separately. However, both terms lead to the same condition (under which the above requirement holds):
$$0 \ll p(x,y) \mbox{ and } 0 \ll p(x,s_i), \mbox{ for some } s_i \in X\backslash \{y\}.$$
In all, the total curvature measure leads to poor performance bounds when there exists some $y\in X$ and $s_i\in X\backslash \{y\}$ so that 
$$0\ll p(x,y) \simeq p(x,s_i) \simeq 1,$$
for many feasible space locations $x\in \Phi$. Evidently, this requirement holds when the ground set $X$ is dense and when the agents have significantly strong overlapping sensing capabilities (i.e., large range $\delta$ and small decay $\lambda$ in \eqref{Eq:SensingFunction}).

One final remark about the total curvature measure is that it requires an evaluation of $H(X)$ and $M(\triangleq \vert X\vert$) evaluations of $H(X\backslash \{y\})$ terms (i.e., for all $y\in X$). In certain coverage applications, this might be ill-defined \cite{Sun2020} and computationally expensive as often $H(S)$ is of the complexity $O(\vert S \vert)$.

\subsection{Greedy Curvature \cite{Conforti1984}}
\label{SubSec:GreedyCurvature}

The \emph{greedy curvature} of \eqref{Eq:SetCoverageProblem} is given by   
\begin{equation}\label{Eq:GreedyCurvatureCoefficientTheory}
    \alpha_g \triangleq \max_{0 \leq i \leq N-1} \left[ \max_{y \in X^i}\left(1 - \frac{\Delta H(y\vert S^i)}{\Delta H(y\vert \emptyset)}\right) \right],
\end{equation}
where $X^i \triangleq \{y: y \in X \backslash S^i, (S^i \cup \{y\}) \in \mathcal{I}^N\}$ (i.e., the set of feasible options at the $(i+1)$\tsup{th} greedy iteration). The corresponding performance bound $\beta_g$ is given by  
\begin{equation}\label{Eq:GreedyCurvatureBoundTheory}
    \beta_g \triangleq 1-\alpha_g\left(1-\frac{1}{N}\right) \leq \frac{H(S^G)}{H(S^*)}. 
\end{equation}

Note that $\beta_g$ is a monotonically decreasing function in $\alpha_g$, and due to the submodularity of $H$, $0 \leq \alpha_g \leq 1$. Consequently, as $\alpha_g \rightarrow 0$, $\beta_g \rightarrow 1$, and on the other hand, as $\alpha_g \rightarrow 1$, $\beta_g \rightarrow \frac{1}{N}$ (which may be worse than $\beta_f$, when $\frac{1}{N} < \beta_f$). Using these facts and \eqref{Eq:GreedyCurvatureCoefficientTheory}, it is easy to see that the improvement in the performance bound is proportional to the magnitude of 
\begin{equation}\label{Eq:GreedyCurvatureGamma}
    \gamma_g \triangleq \min_{0 \leq i \leq N-1} \left[ \min_{y \in X^i}\left(\frac{\Delta H(y\vert S^i)}{\Delta H(y\vert \emptyset)}\right) \right].
\end{equation}
Similar to before, the diminishing returns property of $H$ implies that $\gamma_g$ is large only when 
$\frac{\Delta H(y\vert S^i)}{\Delta H(y\vert \emptyset)} \simeq 1, \forall y\in X^i, i\in \{0,1,2,...,N-1\}$. In other words, similar to the total curvature, the greedy curvature provides a significantly improved performance bound when $H$ is weakly submodular.

In fact, as reported in \cite{Sun2020}, when $H$ is significantly weakly submodular, it can provide better performance bounds even compared to those provided by the total curvature, i.e., $\beta_f \ll \beta_t \leq \beta_g \simeq 1$. This observation can be theoretically justified using \eqref{Eq:TotolCurvatureGamma} and \eqref{Eq:GreedyCurvatureGamma} as follows. Due to submodularity, $\Delta H(y \vert X \backslash \{y\}) \leq \Delta H(y\vert S^i)$ for any $y$ and $S^i$, and thus, $\gamma_t \leq \gamma_g$. This, with weak submodularity of $H$ leads to $\alpha_t \geq \alpha_g \simeq 0$. Now, noticing that the growth of $\beta_g$ is faster as $\alpha_g \rightarrow 0$ compared to that of $\beta_t$ as $\alpha_t \rightarrow 0$, we get $\beta_f \ll \beta_t \leq \beta_g \simeq 1$.

We can follow the same steps and arguments as before to show that such improved performance bounds can only be achieved when the ground set is sparse and/or when the agents have weak sensing capabilities. On the other hand, when the ground set is dense and when the agents have strong sensing capabilities, greedy curvature provides poor performance bounds (often, it may even be worse than $\beta_f$). 

However, compared to the total curvature, greedy curvature has some more redeeming qualities: it is always fully defined, and it can be computed efficiently using only the evaluations of $H$ executed in the greedy algorithm.

\subsection{Elemental Curvature \cite{Wang2016}}\label{SubSec:ElementalCurvature}

The \emph{elemental curvature} of \eqref{Eq:SetCoverageFunction} is given by 
\begin{equation}\label{Eq:ElementalCurvatureCoefficientTheory}
    \alpha_e \triangleq \max_{\substack{(S,y_i,y_j): S \subset X,\\ y_i,y_j \in X \backslash S,\ y_i \neq y_j.}}\left[\frac{\Delta H(y_i \vert S \cup \{y_j\})}{\Delta H(y_i \vert S)}\right].
\end{equation}
The corresponding performance bound $\beta_e$ is given by 
\begin{equation}\label{Eq:ElementalCurvatureBoundTheory}
    \beta_e \triangleq 1-\left(\frac{\alpha_e + \alpha_e^2 + \cdots + \alpha_e^{N-1}}{1 + \alpha_e + \alpha_e^2 + \cdots + \alpha_e^{N-1}}\right)^N \leq \frac{H(S^G)}{H(S^*)}.
\end{equation}

It can be shown that $\beta_e$ is a monotonically decreasing function in $\alpha_e$, and due to the submodularity of $H$, $0 \leq \alpha_e \leq 1$. Consequently, when $\alpha_e \rightarrow 0$, $\beta_e \rightarrow 1$ and when $\alpha_e \rightarrow 1$, $\beta_e \rightarrow \beta_f$ (the latter is unlike $\beta_g$, but similar to $\beta_t$).

Since $H$ is submodular, according to  \cite[Prop. 2.1]{Nemhauser1978}, for all feasible $(S,y_i,y_j)$ considered in \eqref{Eq:ElementalCurvatureCoefficientTheory}, $\frac{\Delta H(y_i \vert S \cup \{y_j\})}{\Delta H(y_i \vert S)} \leq 1$. Therefore, based on \eqref{Eq:ElementalCurvatureCoefficientTheory}, whenever there exists some feasible $(S,y_i,y_j)$ such that  $\frac{\Delta H(y_i \vert S \cup \{y_j\})}{\Delta H(y_i \vert S)} \simeq 1$, i.e. when $H$ is weakly submodular (or, equivalently, closer to being modular) in that region, the elemental curvature measure will provide poor performance bounds (closer to $\beta_f$). This modularity argument is also evident from considering a simplified case of condition $\frac{\Delta H(y_i \vert S \cup \{y_j\})}{\Delta H(y_i \vert S)} \simeq 1$ assuming $S=\emptyset$, as it leads to
$$
H(\{y_i,y_j\}) \simeq H(\{y_j\}) + H(\{y_i\}), \quad \mbox{for some } y_i,y_j\in X
$$
which holds whenever $H$ is modular.

As we observed before, the coverage function $H$ shows such modular behaviors (leading to poor performance bounds $\beta_e\simeq \beta_f$) when the ground set $X$ is very sparse and/or when agents have significantly weak non-overlapping sensing capabilities. It is worth highlighting that this particular behavior of elemental curvature contrasts from that of the previously discussed total curvature and greedy curvature - where weakly submodular scenarios (with agents having weak sensing capabilities) lead to significantly improved performance bounds $\beta_f \ll \beta_t \leq \beta_g \simeq 1$.

On the other hand, the elemental curvature provides an improved performance bound when $\frac{\Delta H(y_i \vert S \cup \{y_j\})}{\Delta H(y_i \vert S)} \ll 1$ over all feasible $(S,y_i,y_j)$ considered in \eqref{Eq:ElementalCurvatureCoefficientTheory}. To further interpret this condition, we need to consider the corresponding marginal detection function \eqref{Eq:MarginalDetection} requirement:
\begin{equation}\label{Eq:ConditionForStrongSubmodularity}
\Delta P(x,y_i \vert S \cup \{y_j\}) \ll \Delta P(x, y_i \vert S), \quad \forall (S,y_i,y_j)    
\end{equation}
for a majority of $x\in\Phi$. Since each $\Delta P = \theta \Delta P_J + (1-\theta)\Delta P_M$ where $\theta \in [0,1]$, let us first consider the requirement \eqref{Eq:ConditionForStrongSubmodularity} with respect to the $\Delta P_J$ (i.e., when $\theta = 1$ in \eqref{Eq:SetDetectionFunction}) using \eqref{Eq:Th:SetCoverageSubmodularityStep1}:
\begin{align}
    &\Delta P_J(x,y_i \vert S \cup \{y_j\}) \ll \Delta P_J(x, y_i \vert S) \nonumber\\
    &\iff p(x,y_i)\prod_{s_i\in S\cup\{y_j\}}(1-p(x,s_i))\ll 
    p(x,y_i)\prod_{s_i\in S}(1-p(x,s_i))\nonumber\\
    &\iff 0 \ll p(x,y_i)p(x,y_j)\prod_{s_i\in S}(1-p(x,s_i)). \nonumber
\end{align}
Clearly, this condition holds if for all feasible $(S,y_i,y_j)$, 
\begin{equation}\label{Eq:ConditionForStrongSubmodularity1}
    0 \ll p(x,y_i) \simeq p(x,y_j) \simeq 1 \mbox{ with } 0 \simeq p(x,s_i) \ll 1
\end{equation}
for some $s_i\in S$ over many feasible space locations $x\in\Phi$.

Now, let us consider the requirement \eqref{Eq:ConditionForStrongSubmodularity} with respect to the $\Delta P_M$ (i.e., when $\theta = 0$ in \eqref{Eq:SetDetectionFunction}) using \eqref{Eq:SetMaxDetection}:
\begin{align}
    &\Delta P_M(x,y_i \vert S \cup \{y_j\}) \ll \Delta P_M(x, y_i \vert S) \nonumber\\
    & \iff 
    \max_{s_i\in S\cup\{y_j,y_i\}} p(x,s_i) - \max_{s_i\in S\cup\{y_j\}} p(x,s_i) \nonumber \\
    &\quad \quad \quad \ll
    \max_{s_i\in S\cup\{y_i\}} p(x,s_i) - \max_{s_i\in S} p(x,s_i)\nonumber\\
    &\iff \max_{s_i\in S\cup\{y_j,y_i\}} p(x,s_i) + \max_{s_i\in S} p(x,s_i) 
    \nonumber\\ \label{Eq:ConditionForStrongSubmodularity2}
    &\quad \quad \quad \ll \max_{s_i\in S\cup\{y_i\}} p(x,s_i) + \max_{s_i\in S\cup\{y_j\}} p(x,s_i).
\end{align}
For notational convenience, let $P_S \triangleq \max_{s_i\in S} p(x,s_i)$, $P_{y_i}\triangleq p(x,y_i)$ and $P_{y_j}\triangleq p(x,y_j)$. Then, \eqref{Eq:ConditionForStrongSubmodularity2} can be restated as
\begin{align*}
    &\max\{P_S,P_{y_i},P_{y_j}\} + P_S \ll \max\{P_S,P_{y_i}\} + \max\{P_S,P_{y_j}\} \\
    &\iff 
    \max\{2P_S,P_S+P_{y_i},P_S+P_{y_j}\} \\ 
    &\quad \quad \quad \ll \max\{P_S+\max\{P_S,P_{y_j}\},P_{y_i}+\max\{P_S,P_{y_j}\}\}\\
    &\iff 
    \max\{2P_S,P_S+P_{y_i},P_S+P_{y_j}\} \\ 
    &\quad \quad \quad  \ll \max\{\max\{2P_S,P_S+P_{y_j}\},\max\{P_S+P_{y_i},P_{y_i}+P_{y_j}\}\}\\
    &\iff 
    \max\{2P_S,P_S+P_{y_i},P_S+P_{y_j}\} \\ 
    &\quad \quad \quad  \ll \max\{\max\{2P_S,P_S+P_{y_j},P_S+P_{y_i}\},P_{y_i}+P_{y_j}\}\\
    &\iff 
    0 \ll \max\{0,P_{y_i}+P_{y_j} - \max\{2P_S,P_S+P_{y_j},P_S+P_{y_i}\}\}\\
    &\iff \max\{2P_S,P_S+P_{y_j},P_S+P_{y_i}\} \ll  P_{y_i}+P_{y_j}\\
    &\iff P_S + \max\{P_S,P_{y_j},P_{y_i}\} \ll  P_{y_i}+P_{y_j}    
\end{align*}
Since $P_{y_i}$ and $P_{y_j}$ are interchangeable in the above expression, let us denote $P_y \triangleq P_{y_i} \simeq P_{y_j}$. This makes the above condition:
\begin{align}
&P_S + \max\{P_S,P_{y}\} \ll  2P_{y} \nonumber \\
&\iff     
\max\{2P_S-2P_y,P_S-P_y\} \ll 0  \iff P_S \ll P_Y  \nonumber \\
&\iff \max_{s_i \in S} p(x,s_i) \ll p(x,y_i) \simeq p(x,y_j), \nonumber
\end{align}
which leads to the same condition obtained in \eqref{Eq:ConditionForStrongSubmodularity1}.

In all, the elemental curvature measure leads to significantly improved performance bounds when for all $(S,y_i,y_j)$ such that $S \subset X, y_i, y_j \in X\backslash S$ and $y_i \neq y_j$,  
$$
0 \simeq p(x,s_i) \ll p(x,y_i) \simeq p(x,y_j) \simeq 1
$$
for some $s_i \in S$ over many feasible space locations $x\in\Phi$. Clearly, this requirement holds when the ground set $X$ is dense and when the agents have significantly strong overlapping sensing capabilities (i.e., large range $\delta$ and small decay $\lambda$ in \eqref{Eq:SensingFunction}).

Finally, note that the evaluation of the elemental curvature $\alpha_e$ \eqref{Eq:ElementalCurvatureCoefficientTheory} is computationally expensive (even compared to the total curvature) as it involves solving a set function maximization problem (notice the set variable $S$ in \eqref{Eq:ElementalCurvatureCoefficientTheory}). However, as shown in \cite{Sun2019}, there may be special structural properties that can be exploited to obtain at least an upper bound on $\alpha_e$, leading to a lower bound on $\beta_e$ - which would still be a valid performance bound for the optimal coverage problem \eqref{Eq:SetCoverageProblem}. The following proposition serves this purpose.

\begin{proposition}
\label{Pr:ElementalCurvatureBound}
An upper-bound for the elemental curvature $\alpha_e$ in \eqref{Eq:ElementalCurvatureCoefficientTheory} is given by 
\begin{equation}\label{Eq:Pr:ElementalCurvatureBound}
\alpha_e \leq \bar{\alpha}_e \triangleq
1-\left(\min_{\substack{(y_i,y_j,x):y_i\in X, y_j \in X\backslash\{y_i\},\\x\in \Phi, p(x,y_i) \neq 0}} p(x,y_j)\right)\mb{1}_{\{\theta =1\}}.
\end{equation}
\end{proposition}
\begin{proof}
Note that, for any $A,B$ with $A\cap B = \emptyset$, we can write 
$\Delta H(A \vert B) = \int_\Phi R(x) \theta \Delta P_J(x, A \vert B) dx + \int_\Phi R(x) (1-\theta) \Delta P_M(x,A \vert B) dx,$
where each term in the right-hand side (RHS) is strictly positive. Using this fact, $\alpha_e$ in \eqref{Eq:ElementalCurvatureCoefficientTheory} can be upper-bounded by:
\begin{align}
\alpha_e \leq \max_{y_i,y_j,S}\max &\left\{ \frac{\int_\Phi R(x)\theta \Delta P_J(x,y_i\vert S\cup \{y_j\})dx}
{\int_\Phi R(x)\theta \Delta P_J(x,y_i\vert S)dx},\nonumber \right.\\
\label{Eq:Pr:ElementalCurvatureBoundStep1}
&\left. 
\frac{\int_\Phi R(x)(1-\theta) \Delta P_M(x,y_i\vert S\cup \{y_j\})dx}
{\int_\Phi R(x)(1-\theta) \Delta P_M(x,y_i\vert S)dx} \right\}.    
\end{align}

Let us now define subsets of $\Phi$ over which the above inner fraction denominator integrands are non-zero: 
$$\Phi_{y_i,S}^J \triangleq \{x:x\in\Phi,\Delta P_J(x, y_i \vert S)\neq 0\},$$  
$$\Phi_{y_i,S}^M \triangleq \{x:x\in\Phi,\Delta P_M(x, y_i \vert S)\neq 0\}.$$ Using the definitions of the marginal detection functions (and our earlier notations $P_S, P_{y_i}$ and $P_{y_j}$):
\begin{equation}
\begin{aligned}
\Delta P_J(x,y_i\vert S\cup \{y_j\}) =&\  
p(x,y_i)(1-p(x,y_j))\prod_{s_i\in S}(1-p(x,s_i)),\\
\Delta P_J(x,y_i\vert S) =&\  
p(x,y_i)\prod_{s_i\in S}(1-p(x,s_i)),\\
\Delta P_M(x,y_i\vert S\cup \{y_j\}) =&\ \max\{0,P_{y_i}-\max\{P_S,P_{y_j}\}\},\\
\Delta P_M(x,y_i\vert S) =&\ \max\{0,P_{y_i}-P_S\}\},\\
\end{aligned}
\end{equation}
it is easy to see that, for any $x\not\in\Phi_{y_i,S}^J$, $\Delta P_J(x,y_i\vert S\cup \{y_j\}) = 0$, and for any $x\not\in\Phi_{y_i,S}^M$, $\Delta P_M(x,y_i\vert S\cup \{y_j\}) = 0$ (the latter is due to $P_{y_i}<P_S \implies P_{y_i} < P_S \leq \max\{P_S,P_{y_j}\}$). This fact enables restricting the integrals in the two fractions in \eqref{Eq:Pr:ElementalCurvatureBoundStep1} respectively to the sets $\Phi_{y_i,S}^J$ and $\Phi_{y_i,S}^M$, leading to:  
\begin{align}
\alpha_e \leq \max_{y_i,y_j,S}\max \left\{ \frac{\int_{\Phi_{y_i,S}^J} R(x)\theta \Delta P_J(x,y_i\vert S\cup \{y_j\})dx}
{\int_{\Phi_{y_i,S}^J} R(x)\theta \Delta P_J(x,y_i\vert S)dx},\nonumber \right.\\\left. 
\frac{\int_{\Phi_{y_i,S}^M} R(x)(1-\theta) \Delta P_M(x,y_i\vert S\cup \{y_j\})dx}
{\int_{\Phi_{y_i,S}^M} R(x)(1-\theta) \Delta P_M(x,y_i\vert S)dx} \right\}\nonumber\\
\leq 
\max_{y_i,y_j,S}\max \left\{\max_{x\in \Phi_{y_i,S}^J} \frac{\Delta P_J(x,y_i\vert S\cup \{y_j\})}
{\Delta P_J(x,y_i\vert S)},\nonumber \right.\nonumber\\
\left. 
\max_{x\in \Phi_{y_i,S}^M} \frac{\Delta P_M(x,y_i\vert S\cup \{y_j\})}
{\Delta P_M(x,y_i\vert S)} \right\}\nonumber\\
\leq 
\max_{y_i,y_j,S}\max \left\{\max_{x\in \Phi_{y_i,S}^J} (1-p(x,y_j)),\nonumber \right.\\
\label{Eq:Pr:ElementalCurvatureBoundStep2}
\left. 
\max_{x\in \Phi_{y_i,S}^M}\frac{\max\{0,P_{y_i}-\max\{P_S,P_{y_j}\}\}}
{P_{y_i}-P_S} \right\}
\end{align}

Let us now consider the inner second fraction term. Note that, for any $x\in \Phi_{y_i,S}^M$, $P_{y_i}>P_S$, which implies three possibilities for $P_{y_j}$: 
(1) $P_{y_i}>P_S \geq P_{y_j}$;
(2) $P_{y_i} > P_{y_j} >P_S$; or 
(3) $P_{y_j} \geq P_{y_i} >P_S$. Based on this, we  define three mutually exclusive sub-regions of $\Phi_{y_i,S}^M = \Phi_{y_i,S}^{M,1} \cup \Phi_{y_i,S}^{M,2} \cup \Phi_{y_i,S}^{M,3}$ as 
\begin{equation*}
\begin{aligned}
\Phi_{y_i,S}^{M,1} \triangleq \{x:p(x,y_i) > \max_{s_i\in S}p(x,s_i) \geq p(x,y_j)\},\\
\Phi_{y_i,S}^{M,2} \triangleq \{x:p(x,y_i) > p(x,y_j) > \max_{s_i\in S}p(x,s_i)\},\\
\Phi_{y_i,S}^{M,3} \triangleq \{x:p(x,y_j) \geq p(x,y_i) > \max_{s_i\in S}p(x,s_i)\}.
\end{aligned}    
\end{equation*}
Using these sub-regions, the inner second fraction term in \eqref{Eq:Pr:ElementalCurvatureBoundStep2} can be simplified as  
\begin{equation*}
\frac{\max\{0,P_{y_i}-\max\{P_S,P_{y_j}\}\}}
{P_{y_i}-P_S} = 
\begin{cases}
    \frac{P_{y_i}-P_S}{P_{y_i}-P_S}=1,\quad &x \in \Phi_{y_i,S}^{M,1}\\
    \frac{P_{y_i}-P_{y_j}}{P_{y_i}-P_S}<1,\quad &x \in \Phi_{y_i,S}^{M,2}\\
    0.\quad  &x \in \Phi_{y_i,S}^{M,3}
\end{cases}
\end{equation*}
Therefore, as one can always select $y_i, y_j$ and $S$ such that $\Phi_{y_i,S}^{M,1} \neq \emptyset$, the inner second maximization in \eqref{Eq:Pr:ElementalCurvatureBoundStep2} becomes $1$. Consequently, all outer maximizations in \eqref{Eq:Pr:ElementalCurvatureBoundStep2} becomes $1$. Thus, whenever we use $\theta \neq 1$ in \eqref{Eq:SetDetectionFunction}, the above approach to bounding $\alpha_e$ leads to the trivial bound $\alpha_e \leq 1$. 

Note, however, that, if $\theta = 1$, we can omit the inner second fraction term in \eqref{Eq:Pr:ElementalCurvatureBoundStep2} entirely. This leads to 
\begin{equation*}
\alpha_e \leq 
\max_{y_i,y_j,S} \max_{x\in \Phi_{y_i,S}^J} (1-p(x,y_j)) 
= 1-\min_{y_i,y_j,S} \min_{x\in \Phi_{y_i,S}^J} p(x,y_j).
\end{equation*}
Finally, recall that $x\in\Phi_{y_i,S}^J$ if and only if $p(x,y_i) \neq 0$ and $p(x,s_i) \neq 1$ for some $s_i\in S$. Note that, based on \eqref{Eq:SensingFunction}, for any $s_i\in S$, $p(x,s_i) \neq 1 \iff x\neq s_i$. Therefore, the condition $p(x,s_i) \neq 1$ for some $s_i\in S$ holds for any choice of $S$. Thus, $x\in\Phi_{y_i,S}^J$ if and only if $p(x,y_i) \neq 0$, leading to    
\begin{align*}
\alpha_e \leq&\  
1-\min_{y_i,y_j} \min_{x: p(x,y_i) \neq 0} p(x,y_j)\\
=&\ 1-\min_{\substack{(y_i,y_j,x):y_i\in X, y_j \in X\backslash\{y_i\},\\ 
x\in \Phi, p(x,y_i) \neq 0}} p(x,y_j).
\end{align*}
This completes the proof.
\end{proof}

\begin{remark}
Evaluating the elemental curvature upper-bound $\bar{\alpha}_e$ proposed in \eqref{Eq:Pr:ElementalCurvatureBound} is significantly more computationally efficient compared to evaluating the original elemental curvature metric $\alpha_e$ as defined in \eqref{Eq:ElementalCurvatureCoefficientTheory}. 
A valid performance bound for the greedy solution then can be obtained by using the computed $\bar{\alpha}_e$ value to substitute for $\alpha_e$ in \eqref{Eq:ElementalCurvatureBoundTheory}. 
\end{remark}

\begin{remark}\label{Rm:ElementalCurvatureBoundTriviality}
The proposed elemental curvature upper-bound $\bar{\alpha}_e$ proposed in \eqref{Eq:Pr:ElementalCurvatureBound} becomes trivial (i.e., $\bar{\alpha}_e=1$) under two scenarios. The first scenario is when we can place two agents in the mission space ground set (i.e., find $y_i,y_j\in X$) such that there is no overlapping in their sensing regions (i.e., when there exists $x\in \Omega$ such that $p(x,y_i)\neq 0$ but $p(x,y_j)=0$). The second scenario is when the max detection function is used in the coverage objective (i.e., when $\theta \neq 1$ in \eqref{Eq:SetDetectionFunction}). Note, however, that the lack of a computationally efficient non-trivial $\bar{\alpha}_e \neq 1$ does not guarantee $\alpha_e = 1$ - as it is just a result of our particular approach used to establish an upper-bound for the elemental curvature $\alpha_e$. Future research is directed towards addressing these challenges.
\end{remark}



\subsection{Partial Curvature \cite{Liu2018}}
\label{SubSec:PartialCurvature}

The \emph{partial curvature} of \eqref{Eq:SetCoverageProblem} is given by
\begin{equation}\label{Eq:PartialCurvatureCoefficientTheory}
    \alpha_p = 
    \max_{\substack{(S,y): y \in S \in \mathcal{I}^N}}\left[1-\frac{\Delta H(y \vert S \backslash \{y\}) }{\Delta H(y \vert \emptyset)}\right].
\end{equation}
The corresponding performance bound $\beta_p$ is given by 
\begin{equation}\label{Eq:PartialCurvatureBoundTheory}
    \beta_p \triangleq \frac{1}{\alpha_p}\left[1-\left(1-\frac{\alpha_p}{N}\right)^N\right] \leq \frac{H(S^G)}{H(S^*)}.
\end{equation}

This partial curvature measure $\alpha_p$ \eqref{Eq:PartialCurvatureCoefficientTheory} provides an alternative to the total curvature measure $\alpha_t$ \eqref{Eq:TotalCurvatureCoefficientTheory}. In particular, it addresses the potentially ill-defined nature of the $H(X)$ term involved in $\alpha_t$ \eqref{Eq:TotalCurvatureCoefficientTheory}. Consequently, $\alpha_p$ can be evaluated when the domain of $H$ constrained, i.e., when $H:\mathcal{I}\rightarrow \R_{\geq0}$ with some $\mathcal{I} \subset 2^X$.

The above $\beta_p$ \eqref{Eq:PartialCurvatureBoundTheory} is only valid under a few additional conditions on $f$, $X$ and $\mathcal{I}^N$ (which are omitted here, but can be found in \cite{Liu2018}). Note that we can directly compare $\alpha_t$ and $\alpha_p$ to conclude regarding the nature of the corresponding performance bounds $\beta_t$ and $\beta_p$, as $\beta_t$ \eqref{Eq:TotalCurvatureBoundTheory} and $\beta_p$ \eqref{Eq:PartialCurvatureBoundTheory} has identical forms. The work in \cite{Liu2018} has established that $\alpha_p \leq \alpha_t$, which implies that $\beta_p \geq \beta_t$, i.e., $\beta_p$ is always tighter than $\beta_t$. 

Note also that, similar to $\beta_t$, $\beta_p$ will provide significantly improved performance bounds (i.e., $\beta_p \simeq 1$) when $H$ is weakly submodular. As observed before, such a scenario occurs when the ground set is sparse and/or agent sensing capabilities are weak. On the other hand, again, similar to $\beta_t$, $\beta_p$ will provide poor performance bounds (i.e., $\beta_p \simeq \beta_f$) when $H$ is strongly submodular. This happens when the ground set is dense, and agent sensing capabilities are strong.

Unfortunately, similar to the elemental curvature $\alpha_e$ \eqref{Eq:ElementalCurvatureCoefficientTheory}, evaluating the partial curvature $\alpha_p$ \eqref{Eq:PartialCurvatureCoefficientTheory} involves solving a set function maximization problem (notice the set variable $Y$ in \eqref{Eq:PartialCurvatureCoefficientTheory}). Therefore, evaluating $\alpha_p$ is much more computationally expensive compared to evaluating $\alpha_t$ or $\alpha_g$. However, like in the case of  $\alpha_e$, we can exploit some special structural properties of the optimal coverage problem to overcome this challenge. 

\begin{proposition}\label{Pr:PartialCurvatureBound}
An upper-bound for the partial curvature $\alpha_p$ in \eqref{Eq:PartialCurvatureCoefficientTheory} is given by  
\begin{equation}\label{Eq:PartialCurvatureBound}
    \alpha_p \leq \bar{\alpha}_p = \frac{1}{\beta_f} \max_{y\in X} \left(1-\frac{\Delta H(y \vert A^G_y)}{H(\{y\})}\right)
\end{equation}
where $A^G_y$ is the greedy solution for the polymatroid maximization (under a uniform matroid constraint) problem
\begin{equation}\label{Eq:PartialCurvatureBound2}
    A_y^* = \underset{\substack{A\subseteq X\backslash \{y\}, \\ \vert A \vert = N-1}}{\arg\max} \left(-\Delta H(y \vert A)+H(\{y\})\right).
\end{equation}
\end{proposition}
\begin{proof}
Using the change of variables $A \triangleq S\backslash \{y\}$ (instead of $S$) in \eqref{Eq:PartialCurvatureCoefficientTheory}, we get 
\begin{align}
\alpha_p =&\  
\max_{\substack{(A,y): y\in X, \\ A \subseteq X \backslash \{y\}, \vert A \vert = N-1}}\left[1-\frac{\Delta H(y \vert A) }{\Delta H(y \vert \emptyset)}\right]\\
=&\ 
\max_{y\in X} \left(\frac{1}{H(\{y\})}\max_{\substack{A \subseteq X \backslash \{y\}\\\vert A \vert = N-1}}\left[H(\{y\})-\Delta H(y \vert A)\right] \right).   
\end{align}
According to Th. \ref{Th:SubmodularityOfMarginalCoverage}, for any $y\in X$, the above inner maximization problem is a polymatroid maximization problem. Therefore, the fundamental performance bound $\beta_f$ in \eqref{Eq:FundamentalPerformanceBound} (now with $N\rightarrow N-1$) is applicable to relate the performance of its optimal solution $A_y^*$ \eqref{Eq:PartialCurvatureBound2} and its greedy solution $A_y^G$. In particular, using the technical result \cite[Lm. 2(b)]{WelikalaJ02021}, this inner maximization can be replaced by an upper-bound to obtain: 
\begin{equation*}
\alpha_p \leq  
\max_{y\in X} \left(\frac{1}{H(\{y\})}\frac{1}{\beta_f}\left[H(\{y\})-\Delta H(y \vert A_y^G)\right] \right),   
\end{equation*}
which leads to the expression in \eqref{Eq:PartialCurvatureBound}.
\end{proof}

Finally, we point out that the upper bound for $\alpha_p$ established in the above proposition can be computed efficiently, and when used in \eqref{Eq:PartialCurvatureBoundTheory}, provides a lower bound to actual partial curvature-based performance bound $\beta_p$.

\begin{remark}
Evaluating the partial curvature upper-bound $\bar{\alpha}_p$ proposed in \eqref{Eq:PartialCurvatureBound} is significantly more computationally efficient compared to evaluating the original partial curvature metric $\alpha_p$ as defined in \eqref{Eq:PartialCurvatureCoefficientTheory}. A valid performance bound for the greedy solution then can be obtained by using the computed $\bar{\alpha}_p$ value to substitute for $\alpha_p$ in \eqref{Eq:PartialCurvatureBoundTheory}. 
\end{remark}

\subsection{Extended Greedy Curvature \cite{WelikalaJ02021}}\label{SubSec:ExtendedGreedyCurvature}

The \emph{extended greedy curvature}, as the name suggests, requires executing some extra greedy iterations in the greedy algorithm (i.e., Alg. \ref{Alg:GreedyAlgorithm}). This is not an issue as Alg. \ref{Alg:GreedyAlgorithm} can be executed beyond $N$ iterations until $M \triangleq \vert X \vert$ iterations - analogous to a scenario where more than $N$ agents are to be deployed to the mission space in a greedy fashion. 

To define the extended greedy curvature, we first need some additional notations. Recall that we used $(S^i,s^i)$ to denote the greedy (set\,,\,element) observed at the $i$\tsup{th} greedy iteration, where $i\in \N_M^0$. Let $m \triangleq \left \lfloor{\frac{M}{N}}\right \rfloor$, and for any $n\in \N_{m-1}^0$, 
\begin{equation}\label{Eq:GreedySolutionBlock}
    S^G_n \triangleq S^{(n+1)N} \backslash S^{nN} = \{s^{nN+1},s^{nN+2},\ldots,s^{nN+N}\}, 
\end{equation}
\begin{equation}\label{Eq:groundSetBlock}
X_n \triangleq X\backslash S^{nN} \ \ \mbox{ and } \ \  
\mathcal{I}_n^N \triangleq \{S:S \subseteq X_n, \vert S \vert \leq N\}.   
\end{equation}
Simply, $S^G_n$ is the $(n+1)$\tsup{th} block of size $N$ greedy agent placements, and $X_n$ is the the set of agent locations remaining after $nN$ greedy iterations. Note that, $S^G_0 = S^G, X_0 = X$ and $\mathcal{I}_0^N = \mathcal{I}^N$. Note also that, for any $n\in\N_{m-1}^0$, the set system $(X_n,\mathcal{I}_n^N)$ is a uniform matroid of rank $N$, and $S^G_n$ is the greedy solution for ${\arg\max}_{S\in\mathcal{I}_n^N} H(S)$.

The extended greedy curvature of \eqref{Eq:SetCoverageProblem} is given by
\begin{equation}\label{Eq:ExtGreedyCurvatureMeasure}
    \alpha_u \triangleq \min_{i \in Q}\ \alpha^i_u,
\end{equation}
where $Q \subseteq \bar{Q} \triangleq \{i\in\N_M: i = nN+1,\ n\in \N_{m-1}^0 \mbox{ or } i = nN,\ n\in \N_m  \mbox{ or } i = M \}$ and  
\begin{equation}\label{Eq:UpperBoundAlphaiForfYStar}
    \alpha^i_u \triangleq 
    \begin{cases}
     H(S^{i-1}) + \underset{S\in\mathcal{I}^N_{(i-1)/N}}{\max} \left[ \sum_{y\in S} \Delta H(y \vert S^{i-1})\right]\\
     \hfill \mbox{if }i=nN+1,\ n\in\N_{m-1}^0,\\
     H(S^{i-N}) + \frac{1}{\beta_f}\left[ H(S^i) - H(S^{i-N}) \right]\\
     \hfill \mbox{if }i=nN,\ n\in\N_m,\\
     H(S^i) 
     \hfill \mbox{if }i=M. 
    \end{cases}
\end{equation}


The performance bound $\beta_u$ corresponding to the extended greedy curvature measure $\alpha_u$ is given by
\begin{equation}\label{Eq:ExtGreedyCurvatureBoundTheory}
\beta_u \triangleq \frac{H(S^G)}{\alpha_u} \leq \frac{H(S^G)}{H(S^*)}. 
\end{equation}

Note that $\bar{Q}$ is a fixed set of greedy iteration indexes. For each $i\in \bar{Q}$, a corresponding $\alpha^i_u$ value can be computed using known byproducts generated during the execution of greedy iterations. Unlike $\bar{Q}$, $Q$ is an arbitrary subset selected from $\bar{Q}$ based on the user preference. For example, one may choose $Q=\{1,N,N+1,2N,2N+1\}$ so that $\alpha_u$ value can be obtained upon executing only $N+1$ extra greedy iterations. Another motivation for selecting a smaller set for $Q$ compared to $\bar{Q}$ may also be the computational cost associated with running extra greedy iterations. 

However, according to \eqref{Eq:ExtGreedyCurvatureBoundTheory}, $\beta_u$ is a monotonically decreasing function in $\alpha_u$, and according to \eqref{Eq:ExtGreedyCurvatureMeasure}, $\alpha_u$ is a monotonically decreasing set function in $Q$. Consequently, the performance bound $\beta_u$ is a \emph{monotone set function} in $Q$, implying that any superset of $Q$ will always provide a better (or at least the same) $\beta_u$ value compared to that obtained from the set $Q$.

In the context of optimal coverage problem \eqref{Eq:SetCoverageProblem}, to identify unique qualities of this extended greedy curvature-based performance bound, let us first consider $\alpha_u^1$ (from \eqref{Eq:UpperBoundAlphaiForfYStar}):
\begin{equation*}
\alpha_u^1 = 
H(S^{0}) + \underset{S\in\mathcal{I}^N_{0}}{\max} \left[ \sum_{y\in S} \Delta H(y \vert S^{0})\right] = 
\underset{S\in\mathcal{I}^N}{\max} \left[ \sum_{y\in S} H(y)\right]     
\end{equation*}
Notice that when the set coverage function $H$ is closer to being modular (i.e., weakly submodular), the above $\alpha_u^1 \simeq H(S^*)$. Through \eqref{Eq:ExtGreedyCurvatureMeasure} and \eqref{Eq:ExtGreedyCurvatureBoundTheory}, this implies that $\frac{H(S^G)}{\beta_u} = \alpha_u  \leq \alpha_u^1 \simeq H(S^*)$, leading to the conclusion $\beta_u \simeq 1$. Therefore, when $H$ is weakly submodular, i.e., when the ground set is sparse and/or agent sensing capabilities are weak, $\beta_u$ provides significantly improved performance bounds. This bahaviour is similar to that of the performance bounds $\beta_t,\beta_g$ and $\beta_p$ discussed before.

Let us now consider $\alpha_u^{2N}$ (from \eqref{Eq:UpperBoundAlphaiForfYStar}, also using $S^N=S^G$)
\begin{equation}
\alpha_u^{2N} = H(S^G) + \frac{1}{\beta_f}\left[H(S^{2N})-H(S^{N})\right].
\end{equation}
Notice that when the set coverage function $H$ is strongly submodular (i.e., when its ``diminishing returns'' property is severe, see Def. \ref{Def:Submodularity}(3)), the greedy coverage level $H(S^i)$ should saturate quickly with the greedy iterations $i$. Consequently, above $\alpha_u^{2N}\simeq H(S^G)$ as $\frac{1}{\beta_f}(H(S^{2N})-H(S^{N}))\simeq 0$. Through \eqref{Eq:ExtGreedyCurvatureMeasure} and \eqref{Eq:ExtGreedyCurvatureBoundTheory}, this implies that $\frac{H(S^G)}{\beta_u} = \alpha_u  \leq \alpha_u^1 \simeq H(S^G)$ leading to the conclusion $\beta_u \simeq 1$. Therefore, when $H$ is strongly submodular, i.e., when the ground set is dense, and agent sensing capabilities are strong, $\beta_u$ provides significantly improved performance bounds. This bahavior is similar to that of the performance bound $\beta_e$ discussed before.

Moreover, in this strong submodular case, as $\beta_f \ll \beta_e \simeq 1$, it is worth noting that the above factor $\frac{1}{\beta_f}$ (originally appearing in \eqref{Eq:UpperBoundAlphaiForfYStar}) can be replaced with the much smaller factor $\frac{1}{\beta_e}$ (compared to $\frac{1}{\beta_f}$). Therefore, this modification leads to further improvements in the performance bound $\beta_u \simeq 1$. 

In all, the extended greedy curvature measure-based performance bound $\beta_u$ is computationally efficient and provides significantly improved performance bounds under both strong and weak submodular scenarios. This behavior contrasts with that of all other performance bounds discussed before.

\section{Discussion}\label{Sec:Discussion}


In this section, we summarize our findings on different curvature measures in the context of multi-agent optimal coverage problems. We start by analyzing the computational complexity of the greedy solution and considered curvature measures for the coverage problems.

\subsection{Computational Complexity}
Let us assume the sensing function $p(x,s_i)$ in \eqref{Eq:SensingFunction} for any $x\in \Phi, s_i \in X$ is of complexity $O(1)$. Consequently, the set detection function $P(x,S)$ in \eqref{Eq:SetDetectionFunction} for any $x \in \Phi, S \subseteq X$ is of complexity $O(\vert S \vert)$. Now, let us assume the coverage integral in \eqref{Eq:SetCoverageFunction} is computed by discretizing the feasible space $\Phi$ into $\bar{M}$ points. Therefore, the set coverage function $H(S)$ in \eqref{Eq:SetCoverageFunction} for any $S\subseteq X$ is of complexity $O(\vert S \vert \bar{M})$.

Now, looking at the $i$\tsup{th} greedy iteration (see Alg. \ref{Alg:GreedyAlgorithm}), it is easy to see that the marginal gain $\Delta H (y\vert S^{i-1})$ is of complexity $O(i\bar{M})$ and $M-i+1$ such marginal gain evaluations are required. Therefore, complexity of the $i$\tsup{th} greedy iteration is $O((M-i+1)i\bar{M})$. As the iteration index $i \in \N_N$ with $N<M$, the overall complexity of the greedy algorithm can be shown to be $O(N^2 M \bar{M})$. It is worth noting here that the complexity of a brute-force evaluation of \eqref{Eq:SetCoverageProblem} is $O(N{M \choose N}\bar{M}) \sim O(M^N \bar{M})$. 

The total curvature $\alpha_t$ in \eqref{Eq:TotalCurvatureCoefficientTheory} requires additional evaluations $H(X)$ and $H(X\backslash \{y\})$ for all $y\in X$. Using this, it can shown that $\alpha_t$ is of complexity $O(M^2\bar{M})$. The greedy curvature  $\alpha_g$ in \eqref{Eq:GreedyCurvatureCoefficientTheory} requires no additional evaluations except for $N$ ratios; thus, it is of complexity $O(N)$. 

The elemental curvature $\alpha_e$ in \eqref{Eq:ElementalCurvatureCoefficientTheory} involves a set function maximization. Omitting lower order $H$ terms, the complexity of $\beta_e$ can be seen as that of evaluating $H(S\cup\{y_i,y_j\})$ for all possible set variables $\{y_i,y_j\}\subset X$ and $S\subset X\backslash \{y_i,y_j\}$. Using the relationships $\sum_{r=0}^M {M \choose r} = 2^M$ and $\sum_{r=0}^M r {M \choose r} = M2^{M-1}$, it can be shown that $\alpha_e$ is of complexity $O(M^3 2^M \bar{M})$. Note, however, that the conservative upper-bound $\bar{\alpha}_e$ proposed for $\alpha_e$ (see \eqref{Eq:Pr:ElementalCurvatureBound}) is of complexity $O(M^2\bar{M})$ as it involves searching over $\{y_i,y_j\}\subset X$ and $x\in\Phi$ (where $X$ and $\Phi$ have been discretized into $M$ and $\bar{M}$ points, respectively).

The partial curvature $\alpha_p$ in \eqref{Eq:PartialCurvatureCoefficientTheory} also involves a set function maximization. Using similar arguments as before, the complexity of $\beta_p$ can be seen as that of evaluating $H(A\cup \{y_i\})$ for all possible $y_i\in X$ and $A \subseteq X\backslash \{y_i\}$ with $\vert A \vert = N-1$. It can be shown that $\alpha_p$ is of the complexity $O(M^2 2^M \bar{M})$ if the constraint  $\vert A \vert = N-1$ is omitted, otherwise $O(M \sum_{r=0}^{N-1} r{M-1 \choose r}  \bar{M}) \sim O(M^N \bar{M})$. Note also that the conservative upper-bound for $\alpha_p$ given in \eqref{Eq:PartialCurvatureBound} is of complexity $O(N^2 M^2 \bar{M})$ (as it requires $M$ separate greedy solution evaluations where each greedy solution $O(N^2 M \bar{M})$).
      
The extended greedy curvature $\alpha_u$ in \eqref{Eq:ExtGreedyCurvatureMeasure}, if evaluated using $nN+1$ additional greedy iterations, where $n\in \N$, its complexity is $O(n^2N^2M\bar{M})$. Note that when the maximum possible extra greedy iterations are used in evaluating $\alpha_u$, its complexity is $O(M^3\bar{M})$.

\begin{table*}[!ht]
\centering
\caption{Characteristics of Different Curvature Metrics in the Context of Optimal Coverage Problem}
\label{Tab:SummaryTable}
\resizebox{\textwidth}{!}{%
\begin{tabular}{|c|c|cccc|c|c|}
\hline
\multirow{5}{*}{\begin{tabular}[c]{@{}c@{}} Curvature \\ Measure \\ $\alpha$ \end{tabular}} &
\multirow{5}{*}{\begin{tabular}[c]{@{}c@{}} Performance \\ Bound \\ $\beta$ \end{tabular}} &
\multicolumn{4}{c|}{Effectiveness of $\beta$ (i.e., $\beta_f \ll \beta \simeq 1$) when: }&
  \multirow{5}{*}{\begin{tabular}[c]{@{}c@{}} Complexity\end{tabular}} &
  \multirow{5}{*}{\begin{tabular}[c]{@{}c@{}} Additional Remarks \\ (Note: Alg. \ref{Alg:GreedyAlgorithm} $\sim O(N^2 M \bar{M})$)  \end{tabular}} \\ \cline{3-6}
 &  & \multicolumn{2}{c|}{\begin{tabular}[c]{@{}c@{}} Agent Sensing \\ Capabilities \eqref{Eq:SensingFunction} \end{tabular}} & \multicolumn{2}{c|}{Denseness of $X$} &  &  \\ \cline{3-6}
 &  & \multicolumn{1}{c|}{\begin{tabular}[c]{@{}c@{}} Low \\ ($\delta \downarrow,\ \lambda \uparrow$) \end{tabular}} & \multicolumn{1}{c|}{\begin{tabular}[c]{@{}c@{}} High \\ ($\delta \uparrow,\ \lambda \downarrow$) \end{tabular}} & \multicolumn{1}{c|}{\begin{tabular}[c]{@{}c@{}} Low \\ ($M \downarrow$) \end{tabular}}    & \multicolumn{1}{c|}{\begin{tabular}[c]{@{}c@{}} High \\ ($M\uparrow$) \end{tabular}}    &  &  \\ \hline
 ``Total'' $\alpha_t$ & $\beta_t$ & \multicolumn{1}{c|}{\cmark} & \multicolumn{1}{c|}{\xmark} & \multicolumn{1}{c|}{\cmark}    &   \xmark  & $O(M^2 \bar{M})$ &  \\ \hline
 ``Greedy'' $\alpha_g$ & $\beta_g$ & \multicolumn{1}{c|}{\cmark} & \multicolumn{1}{c|}{\xmark} & \multicolumn{1}{c|}{\cmark}    &  \xmark    &  $O(N)$ & \begin{tabular}[c]{@{}c@{}} $\beta_g$ is a byproduct of Alg. \ref{Alg:GreedyAlgorithm} \\ When $\beta_f \ll \beta_t \simeq 1$, $\beta_t \leq \beta_g \simeq 1$ \end{tabular} \\ \hline
 ``Elemental'' $\alpha_e$ & $\beta_e$ & \multicolumn{1}{c|}{\xmark} & \multicolumn{1}{c|}{\cmark} & \multicolumn{1}{c|}{\xmark}    &  \cmark    & $O(M^3 2^M \bar{M})$ & Conservative estimate \eqref{Eq:Pr:ElementalCurvatureBound}: $O(M^2\bar{M})$ \\ \hline
 ``Partial'' $\alpha_p$ & $\beta_p$ & \multicolumn{1}{c|}{\cmark} & \multicolumn{1}{c|}{\xmark} & \multicolumn{1}{c|}{\cmark}    &  \xmark    & $O(M^N \bar{M})$ & 
 \begin{tabular}[c]{@{}c@{}} Without set constraints: $O(M^2 2^M \bar{M})$ \\ Conservative estimate \eqref{Eq:PartialCurvatureBound} $O(N^2 M^2 \bar{M})$ \\ In general: $\beta_f \leq \beta_t \leq \beta_p \leq 1$\end{tabular}  \\ \hline
 ``Ext. Greedy'' $\alpha_u$ & $\beta_e$ & \multicolumn{1}{c|}{\cmark} & \multicolumn{1}{c|}{\cmark} & \multicolumn{1}{c|}{\cmark}    &  \cmark    & $O(n^2N^2M\bar{M})$ &  
 \begin{tabular}[c]{@{}c@{}} i.e., with $nN+1$ extra iterations. \\ With $M-N$ extra iterations: $O(M^3\bar{M})$ \end{tabular} \\ \hline
\end{tabular}%
}
\end{table*}

\subsection{Summary}
Our findings on the effectiveness and computational complexity of different curvature-based performance bounds have been summarized in Tab. \ref{Tab:SummaryTable}. 

In terms of effectiveness, we have observed that total, greedy, and partial curvature measures provide significantly improved performance bounds when agents have low sensing capabilities (i.e., high decay $\lambda$ and/or low range $\delta$) and/or when the ground set is sparse (i.e., low $M$). In particular, compared to the total curvature, (1) greedy curvature performs slightly better in such ``weakly submodular'' scenarios, and (2) partial curvature performs slightly better in general. 
Conversely, the elemental curvature measure provides significantly improved performance bounds when agents have strong sensing capabilities (i.e., low decay $\lambda$ and/or high range $\delta$) and/or when the ground set is dense (i.e., high $M$). Most importantly, the extended greedy curvature distinguishes itself by being able to provide significantly improved performance bounds regardless of the nature of agent sensing capabilities or ground set denseness - proving its versatility in a broad range of scenarios.

In terms of computational complexity, the greedy curvature measure is the most efficient as it can be computed directly using the byproducts of the greedy algorithm (it has a complexity $(O(N))$).
The total curvature exhibits a complexity of \(O(M^2\bar{M})\), manageable but higher than greedy curvature. The conservative upper-bound estimate of the elemental curvature has the same complexity $O(M^2\bar{M})$. In contrast, the original elemental and partial curvatures measures have the highest computational complexities, $O(M^3 2^M \bar{M})$ and $O(M^N\bar{M})$ (i.e., without some constraints, $O(M^3 2^M \bar{M})$), respectively. The conservative upper-bound estimate of the partial curvature has a higher complexity \(O(M^2 2^M \bar{M})\) than that for elemental curvature. The complexity of extended greedy curvature is less considerable compared to that of elemental and partial curvature. However, it is of comparable complexity with respect to that of total curvature and conservative upper bound estimates of elemental and partial curvature measures.

To summarize, this review has highlighted three main challenges in using curvature-based performance bounds for optimal coverage problems: 
(1) the inherent dependence of the effectiveness on the strong or weak nature of the submodularity property of the considered optimal coverage problem,  
(2) the computational complexity associated with computing the curvature measures, and 
(3) the technical conditions required for the successful application of curvature-based performance bounds (e.g., see Remark \ref{Rm:ElementalCurvatureBoundTriviality}). Towards addressing these challenges, the recently proposed extended greedy curvature concept \cite{WelikalaJ02021} has shown promising advances. This curvature measure takes a data-driven approach and utilizes only the information observed during a selected number of extra greedy iterations - offering a computationally efficient performance bound without inherent or technical limitations. 

In light of these findings, we believe future research should be directed toward finding more \emph{data-driven curvature measures} (e.g., $\alpha_u$) to directly address computational challenges faced by standard \emph{theoretical curvature measures} (e.g., $\alpha_e,\alpha_p$). However, in such a pursuit, like in any other data-driven technique development, a crucial challenge would be to establish theoretical guarantees/characterizations on its effectiveness/performance. This challenge motivates exploring \emph{hybrid curvature measures} that have elements rooted in both data-driven curvature measures and theoretical curvature measures. 

In a limited sense, the extended greedy curvature measure $\alpha_u$ can be seen as a hybrid curvature measure as it involves a term $\beta_f$ in \eqref{Eq:UpperBoundAlphaiForfYStar} that can be replaced by $\beta_t,\beta_e$ or $\beta_p$ (which are functions of \emph{theoretical curvature measures} $\alpha_t,\alpha_e$ or $\alpha_p$, respectively). On the other hand, the developed computationally efficient upper bounds on theoretical curvature measures using data-driven techniques (e.g., see $\bar{\alpha}_p$ proposed in Prop. \ref{Pr:PartialCurvatureBound}) can also be seen as a hybrid curvature measure. Nevertheless, the complete theoretical implications of such hybrid curvature measures are yet to be studied, not only in the context of optimal coverage problems but also in the context of broader submodular maximization problems.

\section{Case Studies}
\label{Sec:CaseStudies}

\begin{figure*}[!h]
    \centering
    \begin{subfigure}[t]{0.4\columnwidth}
        \centering
        \includegraphics[width=\columnwidth]{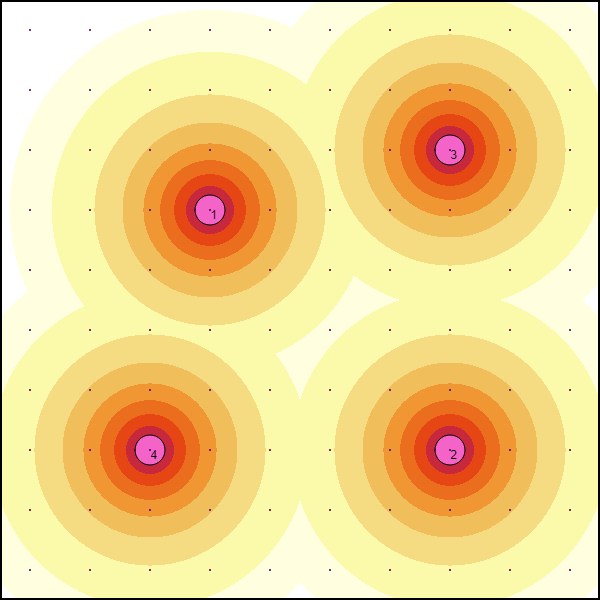}
        \caption{$\theta = 0 \rightarrow \beta_u=0.837$}
    \end{subfigure}%
    \hfill
    \begin{subfigure}[t]{0.4\columnwidth}
        \centering
        \includegraphics[width=\columnwidth]{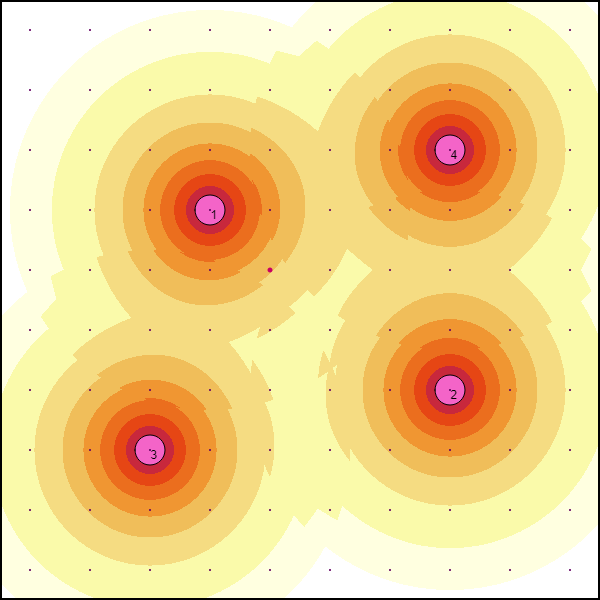}
        \caption{$\theta = 0.3 \rightarrow \beta_u=0.853$}
    \end{subfigure}%
    \hfill
    \begin{subfigure}[t]{0.4\columnwidth}
        \centering
        \includegraphics[width=\columnwidth]{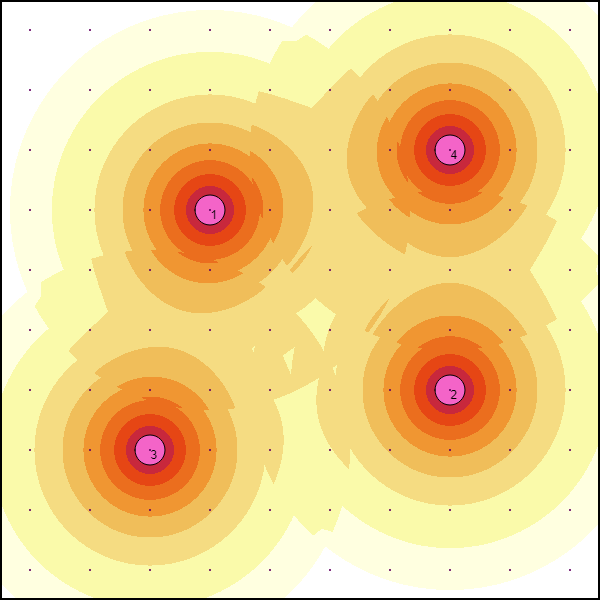}
        \caption{$\theta = 0.5 \rightarrow \beta_u=0.872$}
    \end{subfigure}%
    \hfill
    \begin{subfigure}[t]{0.4\columnwidth}
        \centering
        \includegraphics[width=\columnwidth]{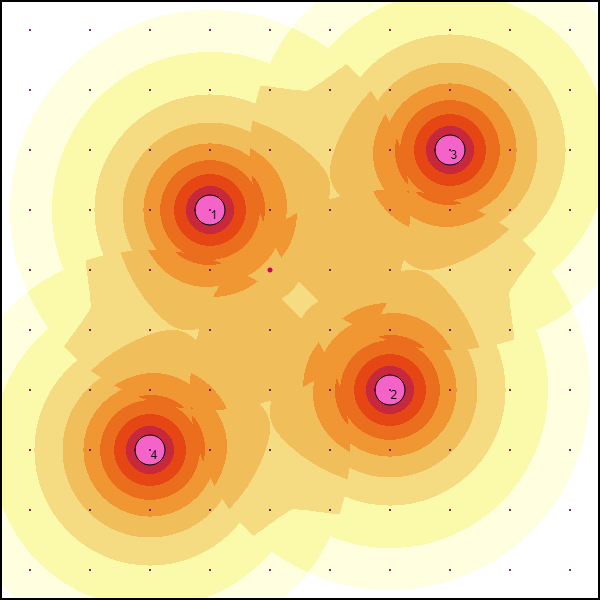}
        \caption{$\theta = 0.7 \rightarrow \beta_u=0.886$}
    \end{subfigure}%
    \hfill
    \begin{subfigure}[t]{0.4\columnwidth}
        \centering
        \includegraphics[width=\columnwidth]{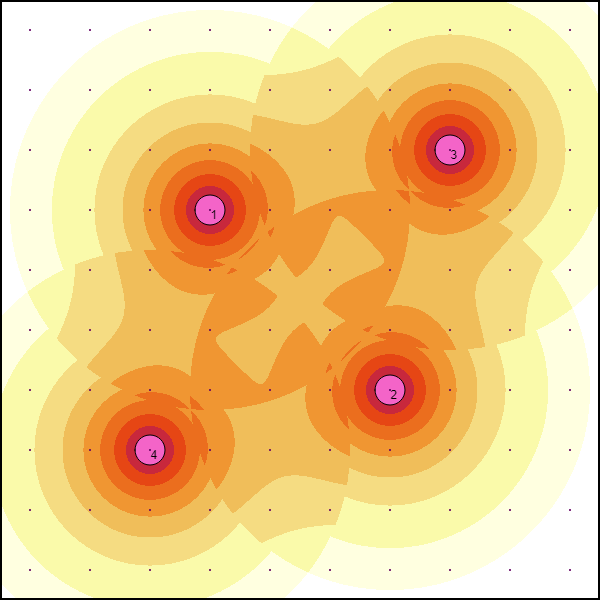}
        \caption{$\theta = 1 \rightarrow \beta_u=0.916$}
    \end{subfigure}%
    \caption{Greedy solutions, coverage level patterns, and the tightest performance bounds observed under different weight parameters $\theta \in [0,1]$ in the Blank mission space with $N=4$ agents with sensing range $\delta=200$ and decay $\lambda = 0.012$.}
    \label{Fig:DetectionFunction}
\end{figure*}

\begin{figure}[!h]
\centering
\begin{minipage}{\columnwidth}
\centering
\captionof{table}{Performance bounds observed under different sensing range $\delta$ values in the Blank mission space with $N=10$ agents with sensing decay $\lambda=0.003$.}
\label{Tab:BlankRange}
\resizebox{0.85\columnwidth}{!}{%
\begin{tabular}{lllllll}
\hline
\multicolumn{7}{|c|}{\textbf{Perf. bounds with respect to $\theta$ at $\delta=35$}} \\ \hline
\multicolumn{1}{|c|}{\textbf{$\theta$}} & \multicolumn{1}{c|}{\textbf{$\beta_f$}} & \multicolumn{1}{c|}{\textbf{$\beta_t$}} & \multicolumn{1}{c|}{\textbf{$\beta_g$}} & \multicolumn{1}{c|}{\textbf{$\beta_e$}} & \multicolumn{1}{c|}{\textbf{$\beta_p$}} & \multicolumn{1}{c|}{\textbf{$\beta_u$}} \\ \hline
\multicolumn{1}{|r|}{0} & \multicolumn{1}{l|}{0.651} & \multicolumn{1}{l|}{0.896} & \multicolumn{1}{l|}{0.834} & \multicolumn{1}{l|}{0.651} & \multicolumn{1}{l|}{0.885} & \multicolumn{1}{l|}{\textbf{1.000}} \\ \hline
\multicolumn{1}{|r|}{0.5} & \multicolumn{1}{l|}{0.651} & \multicolumn{1}{l|}{0.900} & \multicolumn{1}{l|}{0.841} & \multicolumn{1}{l|}{0.651} & \multicolumn{1}{l|}{0.890} & \multicolumn{1}{l|}{\textbf{1.000}} \\ \hline
\multicolumn{1}{|r|}{1} & \multicolumn{1}{l|}{0.651} & \multicolumn{1}{l|}{\textbf{0.904}} & \multicolumn{1}{l|}{\textbf{0.848}} & \multicolumn{1}{l|}{0.651} & \multicolumn{1}{l|}{\textbf{0.895}} & \multicolumn{1}{l|}{{\ul \textbf{1.000}}} \\ \hline
\multicolumn{7}{l}{\cellcolor[HTML]{C0C0C0}} \\ \hline
\multicolumn{7}{|c|}{\textbf{Perf. bounds with respect to $\delta$ at $\theta = 0.5$}} \\ \hline
\multicolumn{1}{|c|}{\textbf{$\delta$}} & \multicolumn{1}{c|}{\textbf{$\beta_f$}} & \multicolumn{1}{c|}{\textbf{$\beta_t$}} & \multicolumn{1}{c|}{\textbf{$\beta_g$}} & \multicolumn{1}{c|}{\textbf{$\beta_e$}} & \multicolumn{1}{c|}{\textbf{$\beta_p$}} & \multicolumn{1}{c|}{\textbf{$\beta_u$}} \\ \hline
\multicolumn{1}{|r|}{35} & \multicolumn{1}{l|}{0.651} & \multicolumn{1}{l|}{\textbf{0.900}} & \multicolumn{1}{l|}{\textbf{0.841}} & \multicolumn{1}{l|}{0.651} & \multicolumn{1}{l|}{\textbf{0.890}} & \multicolumn{1}{l|}{{\ul \textbf{1.000}}} \\ \hline
\multicolumn{1}{|r|}{40} & \multicolumn{1}{l|}{0.651} & \multicolumn{1}{l|}{0.790} & \multicolumn{1}{l|}{0.637} & \multicolumn{1}{l|}{0.651} & \multicolumn{1}{l|}{0.743} & \multicolumn{1}{l|}{\textbf{1.000}} \\ \hline
\multicolumn{1}{|r|}{50} & \multicolumn{1}{l|}{0.651} & \multicolumn{1}{l|}{0.677} & \multicolumn{1}{l|}{0.521} & \multicolumn{1}{l|}{0.651} & \multicolumn{1}{l|}{0.651} & \multicolumn{1}{l|}{\textbf{1.000}} \\ \hline
\multicolumn{1}{|r|}{75} & \multicolumn{1}{l|}{0.651} & \multicolumn{1}{l|}{0.653} & \multicolumn{1}{l|}{0.435} & \multicolumn{1}{l|}{0.651} & \multicolumn{1}{l|}{0.651} & \multicolumn{1}{l|}{\textbf{0.971}} \\ \hline
\multicolumn{1}{|r|}{100} & \multicolumn{1}{l|}{0.651} & \multicolumn{1}{l|}{0.652} & \multicolumn{1}{l|}{0.157} & \multicolumn{1}{l|}{0.651} & \multicolumn{1}{l|}{0.651} & \multicolumn{1}{l|}{\textbf{0.895}} \\ \hline
\multicolumn{1}{|r|}{150} & \multicolumn{1}{l|}{0.651} & \multicolumn{1}{l|}{0.652} & \multicolumn{1}{l|}{0.133} & \multicolumn{1}{l|}{0.651} & \multicolumn{1}{l|}{0.651} & \multicolumn{1}{l|}{\textbf{0.864}} \\ \hline
\multicolumn{1}{|r|}{200} & \multicolumn{1}{l|}{0.651} & \multicolumn{1}{l|}{0.652} & \multicolumn{1}{l|}{0.120} & \multicolumn{1}{l|}{0.651} & \multicolumn{1}{l|}{0.651} & \multicolumn{1}{l|}{\textbf{0.917}} \\ \hline
\multicolumn{1}{|r|}{300} & \multicolumn{1}{l|}{0.651} & \multicolumn{1}{l|}{0.651} & \multicolumn{1}{l|}{0.108} & \multicolumn{1}{l|}{0.651} & \multicolumn{1}{l|}{0.651} & \multicolumn{1}{l|}{\textbf{0.946}} \\ \hline
\multicolumn{1}{|r|}{400} & \multicolumn{1}{l|}{0.651} & \multicolumn{1}{l|}{0.651} & \multicolumn{1}{l|}{0.105} & \multicolumn{1}{l|}{0.651} & \multicolumn{1}{l|}{0.651} & \multicolumn{1}{l|}{\textbf{0.948}} \\ \hline
\multicolumn{1}{|r|}{600} & \multicolumn{1}{l|}{0.651} & \multicolumn{1}{l|}{0.651} & \multicolumn{1}{l|}{0.103} & \multicolumn{1}{l|}{0.651} & \multicolumn{1}{l|}{0.651} & \multicolumn{1}{l|}{\textbf{0.953}} \\ \hline
\multicolumn{1}{|r|}{700} & \multicolumn{1}{l|}{0.651} & \multicolumn{1}{l|}{0.651} & \multicolumn{1}{l|}{0.103} & \multicolumn{1}{l|}{0.651} & \multicolumn{1}{l|}{0.651} & \multicolumn{1}{l|}{\textbf{0.954}} \\ \hline
\multicolumn{1}{|r|}{800} & \multicolumn{1}{l|}{0.651} & \multicolumn{1}{l|}{0.651} & \multicolumn{1}{l|}{0.103} & \multicolumn{1}{l|}{0.651} & \multicolumn{1}{l|}{0.651} & \multicolumn{1}{l|}{\textbf{0.954}} \\ \hline
\multicolumn{7}{|l|}{\cellcolor[HTML]{C0C0C0}} \\ \hline
\multicolumn{7}{|c|}{\textbf{Perf. bounds with respect to $\theta$ at $\delta=800$}}\\ \hline
\multicolumn{1}{|c|}{\textbf{$\theta$}} & \multicolumn{1}{c|}{\textbf{$\beta_f$}} & \multicolumn{1}{c|}{\textbf{$\beta_t$}} & \multicolumn{1}{c|}{\textbf{$\beta_g$}} & \multicolumn{1}{c|}{\textbf{$\beta_e$}} & \multicolumn{1}{c|}{\textbf{$\beta_p$}} & \multicolumn{1}{c|}{\textbf{$\beta_u$}} \\ \hline
\multicolumn{1}{|r|}{0} & \multicolumn{1}{l|}{0.651} & \multicolumn{1}{l|}{\textbf{0.652}} & \multicolumn{1}{l|}{\textbf{0.104}} & \multicolumn{1}{l|}{0.651} & \multicolumn{1}{l|}{0.651} & \multicolumn{1}{l|}{\textbf{0.910}} \\ \hline
\multicolumn{1}{|r|}{0.5} & \multicolumn{1}{l|}{0.651} & \multicolumn{1}{l|}{0.651} & \multicolumn{1}{l|}{0.103} & \multicolumn{1}{l|}{0.651} & \multicolumn{1}{l|}{0.651} & \multicolumn{1}{l|}{\textbf{0.954}} \\ \hline
\multicolumn{1}{|r|}{1} & \multicolumn{1}{l|}{0.651} & \multicolumn{1}{l|}{0.651} & \multicolumn{1}{l|}{0.103} & \multicolumn{1}{l|}{\textbf{0.798}} & \multicolumn{1}{l|}{0.651} & \multicolumn{1}{l|}{{\ul \textbf{0.997}}} \\ \hline
\end{tabular}%
}
\vspace{2mm}
\end{minipage}
\centering
\begin{subfigure}[t]{0.3\columnwidth}
    \centering
    \includegraphics[width=\columnwidth]{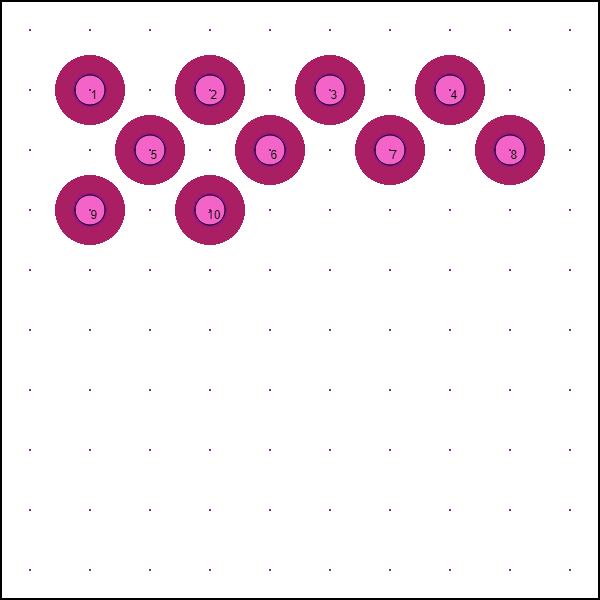}
    \caption{$\delta=35$}
\end{subfigure}%
\hfill
\begin{subfigure}[t]{0.3\columnwidth}
    \centering
    \includegraphics[width=\columnwidth]{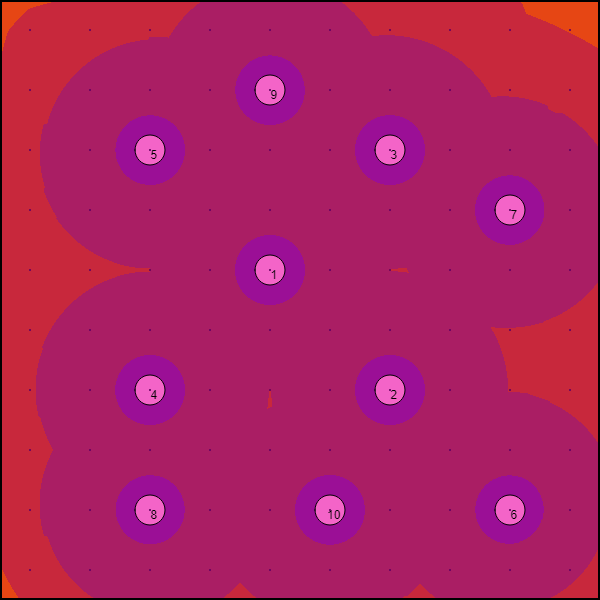}
    \caption{$\delta=400$}
\end{subfigure}%
\hfill
\begin{subfigure}[t]{0.3\columnwidth}
    \centering
    \includegraphics[width=\columnwidth]{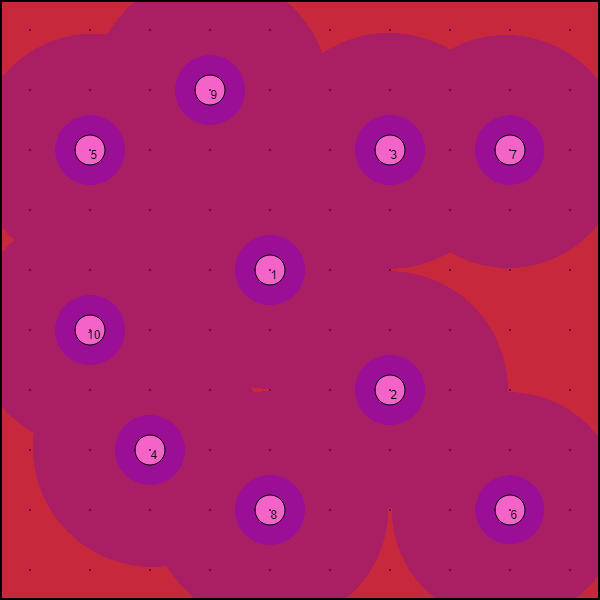}
    \caption{$\delta=800$}
\end{subfigure}%
\caption{Greedy solutions and coverage level patterns observed under different sensing range $\delta$ values considered in Tab. \ref{Tab:BlankRange}.}
\label{Fig:BlankRange}
\end{figure}

\begin{figure}[!h]
\centering
\begin{minipage}{\columnwidth}
\centering
\captionof{table}{Performance bounds observed under different sensing decay $\lambda$ values in the Blank mission space with $N=10$ agents with sensing range $\delta=800$.}
\label{Tab:BlankDecay}
\resizebox{0.85\columnwidth}{!}{%
\begin{tabular}{|lllllll|}
\hline
\multicolumn{7}{|c|}{\textbf{Perf. bounds with respect to $\theta$ at $\lambda=0.05$}} \\ \hline
\multicolumn{1}{|c|}{\textbf{$\theta$}} & \multicolumn{1}{c|}{\textbf{$\beta_f$}} & \multicolumn{1}{c|}{\textbf{$\beta_t$}} & \multicolumn{1}{c|}{\textbf{$\beta_g$}} & \multicolumn{1}{c|}{\textbf{$\beta_e$}} & \multicolumn{1}{c|}{\textbf{$\beta_p$}} & \multicolumn{1}{c|}{\textbf{$\beta_u$}} \\ \hline
\multicolumn{1}{|r|}{0} & \multicolumn{1}{l|}{0.651} & \multicolumn{1}{l|}{0.745} & \multicolumn{1}{l|}{0.595} & \multicolumn{1}{l|}{0.651} & \multicolumn{1}{l|}{0.676} & \textbf{0.943} \\ \hline
\multicolumn{1}{|r|}{0.5} & \multicolumn{1}{l|}{0.651} & \multicolumn{1}{l|}{0.790} & \multicolumn{1}{l|}{0.753} & \multicolumn{1}{l|}{0.651} & \multicolumn{1}{l|}{0.753} & \textbf{0.965} \\ \hline
\multicolumn{1}{|r|}{1} & \multicolumn{1}{l|}{0.651} & \multicolumn{1}{l|}{\textbf{0.840}} & \multicolumn{1}{l|}{\textbf{0.872}} & \multicolumn{1}{l|}{0.651} & \multicolumn{1}{l|}{\textbf{0.829}} & {\ul \textbf{0.992}} \\ \hline
\multicolumn{7}{|l|}{\cellcolor[HTML]{C0C0C0}} \\ \hline
\multicolumn{7}{|c|}{\textbf{Perf. bounds with respect to $\lambda$  at $\theta = 0.5$}} \\ \hline
\multicolumn{1}{|c|}{\textbf{$\lambda$}} & \multicolumn{1}{c|}{\textbf{$\beta_f$}} & \multicolumn{1}{c|}{\textbf{$\beta_t$}} & \multicolumn{1}{c|}{\textbf{$\beta_g$}} & \multicolumn{1}{c|}{\textbf{$\beta_e$}} & \multicolumn{1}{c|}{\textbf{$\beta_p$}} & \multicolumn{1}{c|}{\textbf{$\beta_u$}} \\ \hline
\multicolumn{1}{|r|}{0.05} & \multicolumn{1}{l|}{0.651} & \multicolumn{1}{l|}{\textbf{0.790}} & \multicolumn{1}{l|}{\textbf{0.753}} & \multicolumn{1}{l|}{0.651} & \multicolumn{1}{l|}{\textbf{0.753}} & \textbf{0.965} \\ \hline
\multicolumn{1}{|l|}{0.045} & \multicolumn{1}{l|}{0.651} & \multicolumn{1}{l|}{0.765} & \multicolumn{1}{l|}{0.714} & \multicolumn{1}{l|}{0.651} & \multicolumn{1}{l|}{0.720} & \textbf{0.951} \\ \hline
\multicolumn{1}{|r|}{0.04} & \multicolumn{1}{l|}{0.651} & \multicolumn{1}{l|}{0.739} & \multicolumn{1}{l|}{0.669} & \multicolumn{1}{l|}{0.651} & \multicolumn{1}{l|}{0.686} & \textbf{0.930} \\ \hline
\multicolumn{1}{|r|}{0.035} & \multicolumn{1}{l|}{0.651} & \multicolumn{1}{l|}{0.713} & \multicolumn{1}{l|}{0.617} & \multicolumn{1}{l|}{0.651} & \multicolumn{1}{l|}{0.651} & \textbf{0.901} \\ \hline
\multicolumn{1}{|r|}{0.3} & \multicolumn{1}{l|}{0.651} & \multicolumn{1}{l|}{0.689} & \multicolumn{1}{l|}{0.559} & \multicolumn{1}{l|}{0.651} & \multicolumn{1}{l|}{0.651} & \textbf{0.857} \\ \hline
\multicolumn{1}{|r|}{0.025} & \multicolumn{1}{l|}{0.651} & \multicolumn{1}{l|}{0.670} & \multicolumn{1}{l|}{0.493} & \multicolumn{1}{l|}{0.651} & \multicolumn{1}{l|}{0.651} & \textbf{0.795} \\ \hline
\multicolumn{1}{|r|}{0.02} & \multicolumn{1}{l|}{0.651} & \multicolumn{1}{l|}{0.659} & \multicolumn{1}{l|}{0.416} & \multicolumn{1}{l|}{0.651} & \multicolumn{1}{l|}{0.651} & \textbf{0.705} \\ \hline
\multicolumn{1}{|r|}{0.015} & \multicolumn{1}{l|}{0.651} & \multicolumn{1}{l|}{0.655} & \multicolumn{1}{l|}{0.324} & \multicolumn{1}{l|}{0.651} & \multicolumn{1}{l|}{0.651} & \textbf{0.656} \\ \hline
\multicolumn{1}{|r|}{0.01} & \multicolumn{1}{l|}{0.651} & \multicolumn{1}{l|}{0.653} & \multicolumn{1}{l|}{0.214} & \multicolumn{1}{l|}{0.651} & \multicolumn{1}{l|}{0.651} & \textbf{0.742} \\ \hline
\multicolumn{1}{|r|}{0.005} & \multicolumn{1}{l|}{0.651} & \multicolumn{1}{l|}{0.652} & \multicolumn{1}{l|}{0.118} & \multicolumn{1}{l|}{0.651} & \multicolumn{1}{l|}{0.651} & \textbf{0.912} \\ \hline
\multicolumn{1}{|r|}{0.003} & \multicolumn{1}{l|}{0.651} & \multicolumn{1}{l|}{0.651} & \multicolumn{1}{l|}{0.103} & \multicolumn{1}{l|}{0.651} & \multicolumn{1}{l|}{0.651} & \textbf{0.954} \\ \hline
\multicolumn{1}{|r|}{0.001} & \multicolumn{1}{l|}{0.651} & \multicolumn{1}{l|}{0.651} & \multicolumn{1}{l|}{0.100} & \multicolumn{1}{l|}{0.651} & \multicolumn{1}{l|}{0.651} & {\ul \textbf{0.986}} \\ \hline
\multicolumn{7}{|l|}{\cellcolor[HTML]{C0C0C0}} \\ \hline
\multicolumn{7}{|c|}{\textbf{Perf. bounds with respect to $\theta$ at $\lambda=0.001$}} \\ \hline
\multicolumn{1}{|c|}{\textbf{$\theta$}} & \multicolumn{1}{c|}{\textbf{$\beta_f$}} & \multicolumn{1}{c|}{\textbf{$\beta_t$}} & \multicolumn{1}{c|}{\textbf{$\beta_g$}} & \multicolumn{1}{c|}{\textbf{$\beta_e$}} & \multicolumn{1}{c|}{\textbf{$\beta_p$}} & \multicolumn{1}{c|}{\textbf{$\beta_u$}} \\ \hline
\multicolumn{1}{|r|}{0} & \multicolumn{1}{l|}{0.651} & \multicolumn{1}{l|}{\textbf{0.651}} & \multicolumn{1}{l|}{\textbf{0.101}} & \multicolumn{1}{l|}{0.651} & \multicolumn{1}{l|}{0.651} & \textbf{0.967} \\ \hline
\multicolumn{1}{|r|}{0.5} & \multicolumn{1}{l|}{0.651} & \multicolumn{1}{l|}{0.651} & \multicolumn{1}{l|}{0.100} & \multicolumn{1}{l|}{0.651} & \multicolumn{1}{l|}{0.651} & \textbf{0.986} \\ \hline
\multicolumn{1}{|r|}{1} & \multicolumn{1}{l|}{0.651} & \multicolumn{1}{l|}{0.651} & \multicolumn{1}{l|}{0.100} & \multicolumn{1}{l|}{\textbf{0.998}} & \multicolumn{1}{l|}{0.651} & {\ul \textbf{1.000}} \\ \hline
\end{tabular}%
}
\vspace{3mm}
\end{minipage}
\centering
\begin{subfigure}[t]{0.3\columnwidth}
    \centering
    \includegraphics[width=\columnwidth]{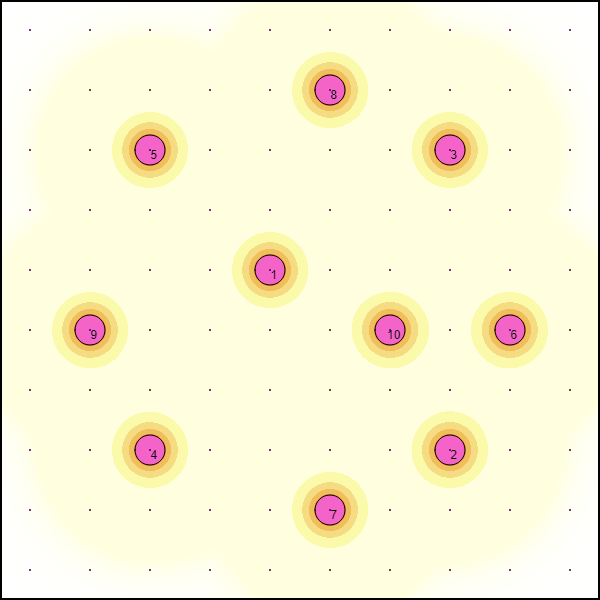}
    \caption{$\lambda = 0.05$}
\end{subfigure}%
\hfill
\begin{subfigure}[t]{0.3\columnwidth}
    \centering
    \includegraphics[width=\columnwidth]{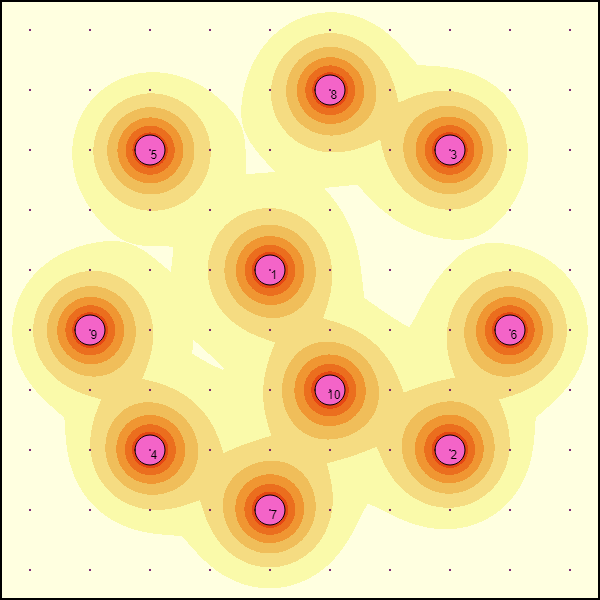}
    \caption{$\lambda = 0.025$}
\end{subfigure}%
\hfill
\begin{subfigure}[t]{0.3\columnwidth}
    \centering
    \includegraphics[width=\columnwidth]{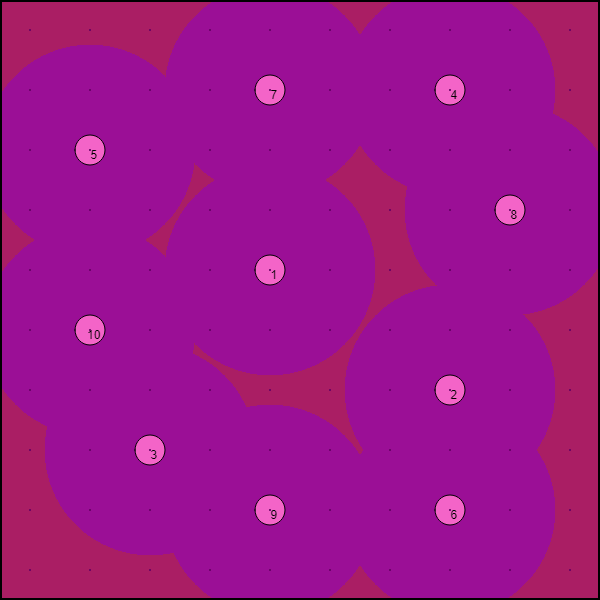}
    \caption{$\lambda = 0.001$}
\end{subfigure}%
\caption{Greedy solutions and coverage level patterns observed under different sensing decay $\lambda$ values considered in Tab. \ref{Tab:BlankDecay}.}
\label{Fig:BlankDecay}
\end{figure}

\begin{figure}[!h]
\centering
\begin{minipage}{\columnwidth}
\centering
\captionof{table}{Performance bounds observed under different sensing range $\delta$ values in the Maze mission space with $N=10$ agents with sensing range $\lambda=0.012$.}
\label{Tab:MazeRange}
\resizebox{0.85\columnwidth}{!}{%
\begin{tabular}{|rrrrrrr|}
\hline
\multicolumn{7}{|c|}{Perf. bounds with respect to $\theta$ at $\delta=50$} \\ \hline
\multicolumn{1}{|c|}{$\theta$} & \multicolumn{1}{c|}{$\beta_f$} & \multicolumn{1}{c|}{$\beta_t$} & \multicolumn{1}{c|}{$\beta_g$} & \multicolumn{1}{c|}{$\beta_e$} & \multicolumn{1}{c|}{$\beta_p$} & \multicolumn{1}{c|}{$\beta_u$} \\ \hline
\multicolumn{1}{|r|}{0} & \multicolumn{1}{r|}{0.651} & \multicolumn{1}{r|}{0.694} & \multicolumn{1}{r|}{0.682} & \multicolumn{1}{r|}{0.651} & \multicolumn{1}{r|}{0.651} & \textbf{0.992} \\ \hline
\multicolumn{1}{|r|}{0.5} & \multicolumn{1}{r|}{0.651} & \multicolumn{1}{r|}{0.718} & \multicolumn{1}{r|}{0.729} & \multicolumn{1}{r|}{0.651} & \multicolumn{1}{r|}{0.651} & \textbf{0.992} \\ \hline
\multicolumn{1}{|r|}{1} & \multicolumn{1}{r|}{0.651} & \multicolumn{1}{r|}{\textbf{0.742}} & \multicolumn{1}{r|}{\textbf{0.775}} & \multicolumn{1}{r|}{0.651} & \multicolumn{1}{r|}{\textbf{0.672}} & \textbf{\underline{0.992}} \\ \hline
\multicolumn{7}{|l|}{\cellcolor[HTML]{C0C0C0}} \\ \hline
\multicolumn{7}{|c|}{Perf. bounds with respect to $\delta$  at $\theta = 0.5$} \\ \hline
\multicolumn{1}{|c|}{$\delta$} & \multicolumn{1}{c|}{$\beta_f$} & \multicolumn{1}{c|}{$\beta_t$} & \multicolumn{1}{c|}{$\beta_g$} & \multicolumn{1}{c|}{$\beta_e$} & \multicolumn{1}{c|}{$\beta_p$} & \multicolumn{1}{c|}{$\beta_u$} \\ \hline
\multicolumn{1}{|r|}{50} & \multicolumn{1}{r|}{0.651} & \multicolumn{1}{r|}{\textbf{0.718}} & \multicolumn{1}{r|}{\textbf{0.729}} & \multicolumn{1}{r|}{0.651} & \multicolumn{1}{r|}{0.651} & \textbf{0.992} \\ \hline
\multicolumn{1}{|r|}{100} & \multicolumn{1}{r|}{0.651} & \multicolumn{1}{r|}{0.660} & \multicolumn{1}{r|}{0.401} & \multicolumn{1}{r|}{0.651} & \multicolumn{1}{r|}{0.651} & \textbf{0.976} \\ \hline
\multicolumn{1}{|r|}{150} & \multicolumn{1}{r|}{0.651} & \multicolumn{1}{r|}{0.657} & \multicolumn{1}{r|}{0.347} & \multicolumn{1}{r|}{0.651} & \multicolumn{1}{r|}{0.651} & \textbf{0.919} \\ \hline
\multicolumn{1}{|r|}{200} & \multicolumn{1}{r|}{0.651} & \multicolumn{1}{r|}{0.657} & \multicolumn{1}{r|}{0.342} & \multicolumn{1}{r|}{0.651} & \multicolumn{1}{r|}{0.651} & \textbf{0.892} \\ \hline
\multicolumn{1}{|r|}{250} & \multicolumn{1}{r|}{0.651} & \multicolumn{1}{r|}{0.656} & \multicolumn{1}{r|}{0.340} & \multicolumn{1}{r|}{0.651} & \multicolumn{1}{r|}{0.651} & \textbf{0.868} \\ \hline
\multicolumn{1}{|r|}{300} & \multicolumn{1}{r|}{0.651} & \multicolumn{1}{r|}{0.656} & \multicolumn{1}{r|}{0.329} & \multicolumn{1}{r|}{0.651} & \multicolumn{1}{r|}{0.651} & \textbf{0.847} \\ \hline
\multicolumn{1}{|r|}{350} & \multicolumn{1}{r|}{0.651} & \multicolumn{1}{r|}{0.656} & \multicolumn{1}{r|}{0.329} & \multicolumn{1}{r|}{0.651} & \multicolumn{1}{r|}{0.651} & \textbf{0.841} \\ \hline
\multicolumn{1}{|r|}{400} & \multicolumn{1}{r|}{0.651} & \multicolumn{1}{r|}{0.656} & \multicolumn{1}{r|}{0.329} & \multicolumn{1}{r|}{0.651} & \multicolumn{1}{r|}{0.651} & \textbf{0.837} \\ \hline
\multicolumn{1}{|r|}{500} & \multicolumn{1}{r|}{0.651} & \multicolumn{1}{r|}{0.656} & \multicolumn{1}{r|}{0.329} & \multicolumn{1}{r|}{0.651} & \multicolumn{1}{r|}{0.651} & \textbf{0.835} \\ \hline
\multicolumn{1}{|r|}{600} & \multicolumn{1}{r|}{0.651} & \multicolumn{1}{r|}{0.656} & \multicolumn{1}{r|}{0.328} & \multicolumn{1}{r|}{0.651} & \multicolumn{1}{r|}{0.651} & \textbf{0.828} \\ \hline
\multicolumn{1}{|r|}{700} & \multicolumn{1}{r|}{0.651} & \multicolumn{1}{r|}{0.656} & \multicolumn{1}{r|}{0.329} & \multicolumn{1}{r|}{0.651} & \multicolumn{1}{r|}{0.651} & \textbf{0.834} \\ \hline
\multicolumn{1}{|r|}{800} & \multicolumn{1}{r|}{0.651} & \multicolumn{1}{r|}{0.656} & \multicolumn{1}{r|}{0.329} & \multicolumn{1}{r|}{0.651} & \multicolumn{1}{r|}{0.651} & \textbf{0.834} \\ \hline
\multicolumn{7}{|l|}{\cellcolor[HTML]{C0C0C0}} \\ \hline
\multicolumn{7}{|c|}{Perf. bounds with respect to $\theta$ at $\delta=800$} \\ \hline
\multicolumn{1}{|c|}{$\theta$} & \multicolumn{1}{c|}{$\beta_f$} & \multicolumn{1}{c|}{$\beta_t$} & \multicolumn{1}{c|}{$\beta_g$} & \multicolumn{1}{c|}{$\beta_e$} & \multicolumn{1}{c|}{$\beta_p$} & \multicolumn{1}{c|}{$\beta_u$} \\ \hline
\multicolumn{1}{|r|}{0} & \multicolumn{1}{r|}{0.651} & \multicolumn{1}{r|}{\textbf{0.660}} & \multicolumn{1}{r|}{0.191} & \multicolumn{1}{r|}{0.651} & \multicolumn{1}{r|}{0.651} & \textbf{0.789} \\ \hline
\multicolumn{1}{|r|}{0.5} & \multicolumn{1}{r|}{0.651} & \multicolumn{1}{r|}{0.656} & \multicolumn{1}{r|}{0.329} & \multicolumn{1}{r|}{0.651} & \multicolumn{1}{r|}{0.651} & \textbf{0.834} \\ \hline
\multicolumn{1}{|r|}{1} & \multicolumn{1}{r|}{0.651} & \multicolumn{1}{r|}{0.652} & \multicolumn{1}{r|}{\textbf{0.460}} & \multicolumn{1}{r|}{0.651} & \multicolumn{1}{r|}{0.651} & \textbf{0.898} \\ \hline
\end{tabular}%
}
\vspace{3mm}
\end{minipage}
\centering
\begin{subfigure}[t]{0.3\columnwidth}
    \centering
    \includegraphics[width=\columnwidth]{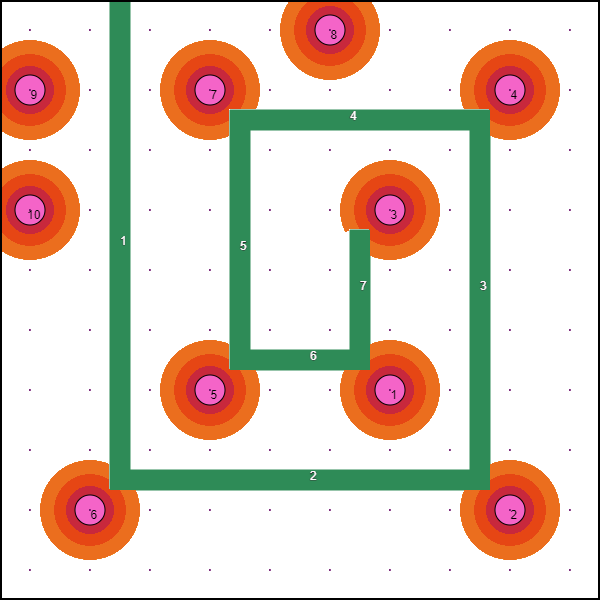}
    \caption{$\delta = 50$}
\end{subfigure}%
\hfill
\begin{subfigure}[t]{0.3\columnwidth}
    \centering
    \includegraphics[width=\columnwidth]{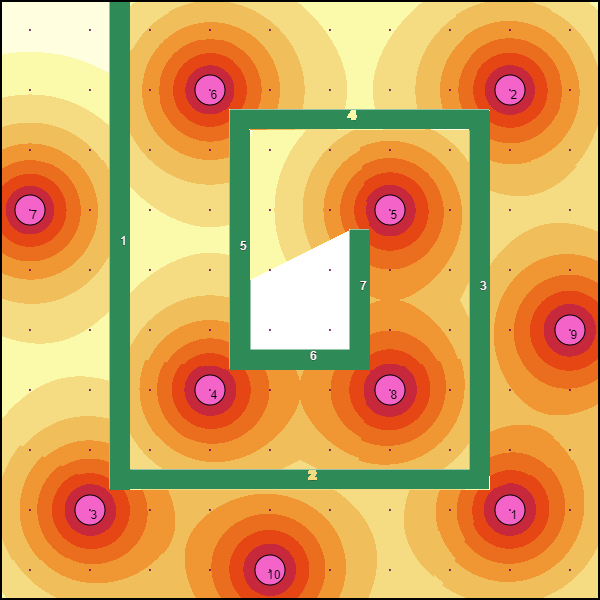}
    \caption{$\delta = 350$}
\end{subfigure}%
\hfill
\begin{subfigure}[t]{0.3\columnwidth}
    \centering
    \includegraphics[width=\columnwidth]{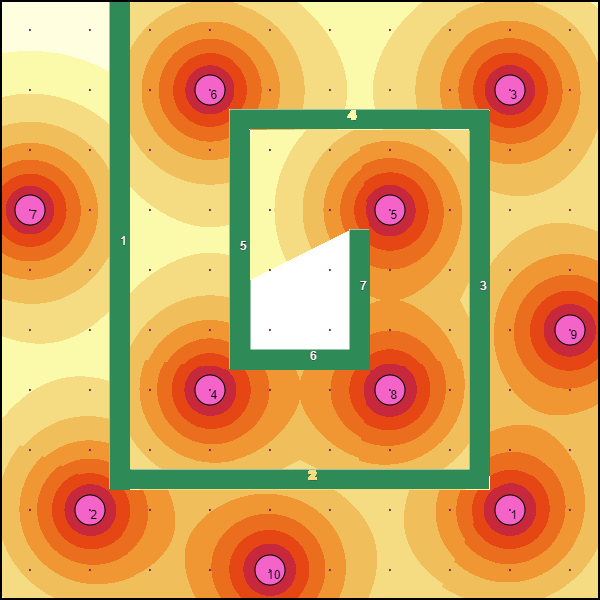}
    \caption{$\delta = 800$}
\end{subfigure}%
\caption{Greedy solutions and coverage level patterns observed under different sensing range $\delta$ values considered in Tab. \ref{Tab:MazeRange}.}
\label{Fig:MazeRange}
\end{figure}

\begin{figure}[!h]
\centering
\begin{minipage}{\columnwidth}
\centering
\captionof{table}{Performance bounds observed under different sensing decay $\lambda$ values in the General mission space with $N=10$ agents with sensing range $\delta=200$.}
\label{Tab:GeneralDecay}
\resizebox{0.85\columnwidth}{!}{%
\begin{tabular}{|rrrrrrr|}
\hline
\multicolumn{7}{|c|}{\textbf{Perf. bounds with respect to $\theta$ at $\lambda=0.05$}} \\ \hline
\multicolumn{1}{|c|}{\textbf{$\theta$}} & \multicolumn{1}{c|}{\textbf{$\beta_f$}} & \multicolumn{1}{c|}{\textbf{$\beta_t$}} & \multicolumn{1}{c|}{\textbf{$\beta_g$}} & \multicolumn{1}{c|}{\textbf{$\beta_e$}} & \multicolumn{1}{c|}{\textbf{$\beta_p$}} & \multicolumn{1}{c|}{\textbf{$\beta_u$}} \\ \hline
\multicolumn{1}{|r|}{0} & \multicolumn{1}{r|}{0.651} & \multicolumn{1}{r|}{0.711} & \multicolumn{1}{r|}{0.579} & \multicolumn{1}{r|}{0.651} & \multicolumn{1}{r|}{0.651} & \textbf{0.926} \\ \hline
\multicolumn{1}{|r|}{0.5} & \multicolumn{1}{r|}{0.651} & \multicolumn{1}{r|}{0.747} & \multicolumn{1}{r|}{0.636} & \multicolumn{1}{r|}{0.651} & \multicolumn{1}{r|}{0.686} & \textbf{0.940} \\ \hline
\multicolumn{1}{|r|}{1} & \multicolumn{1}{r|}{0.651} & \multicolumn{1}{r|}{\textbf{0.786}} & \multicolumn{1}{r|}{\textbf{0.831}} & \multicolumn{1}{r|}{0.651} & \multicolumn{1}{r|}{\textbf{0.751}} & {\ul \textbf{0.972}} \\ \hline
\multicolumn{7}{|l|}{\cellcolor[HTML]{C0C0C0}} \\ \hline
\multicolumn{7}{|c|}{\textbf{Perf. bounds with respect to $\lambda$ at $\theta = 0.5$}} \\ \hline
\multicolumn{1}{|c|}{\textbf{$\lambda$}} & \multicolumn{1}{c|}{\textbf{$\beta_f$}} & \multicolumn{1}{c|}{\textbf{$\beta_t$}} & \multicolumn{1}{c|}{\textbf{$\beta_g$}} & \multicolumn{1}{c|}{\textbf{$\beta_e$}} & \multicolumn{1}{c|}{\textbf{$\beta_p$}} & \multicolumn{1}{c|}{\textbf{$\beta_u$}} \\ \hline
\multicolumn{1}{|r|}{0.05} & \multicolumn{1}{r|}{0.651} & \multicolumn{1}{r|}{\textbf{0.747}} & \multicolumn{1}{r|}{\textbf{0.636}} & \multicolumn{1}{r|}{0.651} & \multicolumn{1}{r|}{\textbf{0.686}} & \textbf{0.940} \\ \hline
\multicolumn{1}{|r|}{0.045} & \multicolumn{1}{r|}{0.651} & \multicolumn{1}{r|}{0.729} & \multicolumn{1}{r|}{0.600} & \multicolumn{1}{r|}{0.651} & \multicolumn{1}{r|}{0.660} & \textbf{0.920} \\ \hline
\multicolumn{1}{|r|}{0.04} & \multicolumn{1}{r|}{0.651} & \multicolumn{1}{r|}{0.710} & \multicolumn{1}{r|}{0.562} & \multicolumn{1}{r|}{0.651} & \multicolumn{1}{r|}{0.651} & \textbf{0.901} \\ \hline
\multicolumn{1}{|r|}{0.035} & \multicolumn{1}{r|}{0.651} & \multicolumn{1}{r|}{0.693} & \multicolumn{1}{r|}{0.520} & \multicolumn{1}{r|}{0.651} & \multicolumn{1}{r|}{0.651} & \textbf{0.878} \\ \hline
\multicolumn{1}{|r|}{0.03} & \multicolumn{1}{r|}{0.651} & \multicolumn{1}{r|}{0.677} & \multicolumn{1}{r|}{0.476} & \multicolumn{1}{r|}{0.651} & \multicolumn{1}{r|}{0.651} & \textbf{0.840} \\ \hline
\multicolumn{1}{|r|}{0.025} & \multicolumn{1}{r|}{0.651} & \multicolumn{1}{r|}{0.666} & \multicolumn{1}{r|}{0.433} & \multicolumn{1}{r|}{0.651} & \multicolumn{1}{r|}{0.651} & \textbf{0.808} \\ \hline
\multicolumn{1}{|r|}{0.02} & \multicolumn{1}{r|}{0.651} & \multicolumn{1}{r|}{0.658} & \multicolumn{1}{r|}{0.382} & \multicolumn{1}{r|}{0.651} & \multicolumn{1}{r|}{0.651} & \textbf{0.748} \\ \hline
\multicolumn{1}{|r|}{0.015} & \multicolumn{1}{r|}{0.651} & \multicolumn{1}{r|}{0.655} & \multicolumn{1}{r|}{0.320} & \multicolumn{1}{r|}{0.651} & \multicolumn{1}{r|}{0.651} & \textbf{0.686} \\ \hline
\multicolumn{1}{|r|}{0.01} & \multicolumn{1}{r|}{0.651} & \multicolumn{1}{r|}{0.653} & \multicolumn{1}{r|}{0.243} & \multicolumn{1}{r|}{0.651} & \multicolumn{1}{r|}{0.651} & \textbf{0.716} \\ \hline
\multicolumn{1}{|r|}{0.005} & \multicolumn{1}{r|}{0.651} & \multicolumn{1}{r|}{0.652} & \multicolumn{1}{r|}{0.154} & \multicolumn{1}{r|}{0.651} & \multicolumn{1}{r|}{0.651} & \textbf{0.820} \\ \hline
\multicolumn{1}{|r|}{0.003} & \multicolumn{1}{r|}{0.651} & \multicolumn{1}{r|}{0.652} & \multicolumn{1}{r|}{0.135} & \multicolumn{1}{r|}{0.651} & \multicolumn{1}{r|}{0.651} & \textbf{0.897} \\ \hline
\multicolumn{1}{|r|}{0.001} & \multicolumn{1}{r|}{0.651} & \multicolumn{1}{r|}{0.651} & \multicolumn{1}{r|}{0.107} & \multicolumn{1}{r|}{0.651} & \multicolumn{1}{r|}{0.651} & {\ul \textbf{0.968}} \\ \hline
\multicolumn{7}{|l|}{\cellcolor[HTML]{C0C0C0}} \\ \hline
\multicolumn{7}{|c|}{\textbf{Perf. bounds with respect to $\theta$ at $\lambda=0.001$}} \\ \hline
\multicolumn{1}{|c|}{\textbf{$\theta$}} & \multicolumn{1}{c|}{\textbf{$\beta_f$}} & \multicolumn{1}{c|}{\textbf{$\beta_t$}} & \multicolumn{1}{c|}{\textbf{$\beta_g$}} & \multicolumn{1}{c|}{\textbf{$\beta_e$}} & \multicolumn{1}{c|}{\textbf{$\beta_p$}} & \multicolumn{1}{c|}{\textbf{$\beta_u$}} \\ \hline
\multicolumn{1}{|r|}{0} & \multicolumn{1}{r|}{0.651} & \multicolumn{1}{r|}{\textbf{0.651}} & \multicolumn{1}{r|}{0.103} & \multicolumn{1}{r|}{0.651} & \multicolumn{1}{r|}{0.651} & \textbf{0.964} \\ \hline
\multicolumn{1}{|r|}{0.5} & \multicolumn{1}{r|}{0.651} & \multicolumn{1}{r|}{0.651} & \multicolumn{1}{r|}{0.107} & \multicolumn{1}{r|}{0.651} & \multicolumn{1}{r|}{0.651} & \textbf{0.968} \\ \hline
\multicolumn{1}{|r|}{1} & \multicolumn{1}{r|}{0.651} & \multicolumn{1}{r|}{0.651} & \multicolumn{1}{r|}{\textbf{0.109}} & \multicolumn{1}{r|}{0.651} & \multicolumn{1}{r|}{0.651} & {\ul \textbf{0.975}} \\ \hline
\end{tabular}}
\vspace{3mm}
\end{minipage}
\centering
\begin{subfigure}[t]{0.3\columnwidth}
    \centering
    \includegraphics[width=\columnwidth]{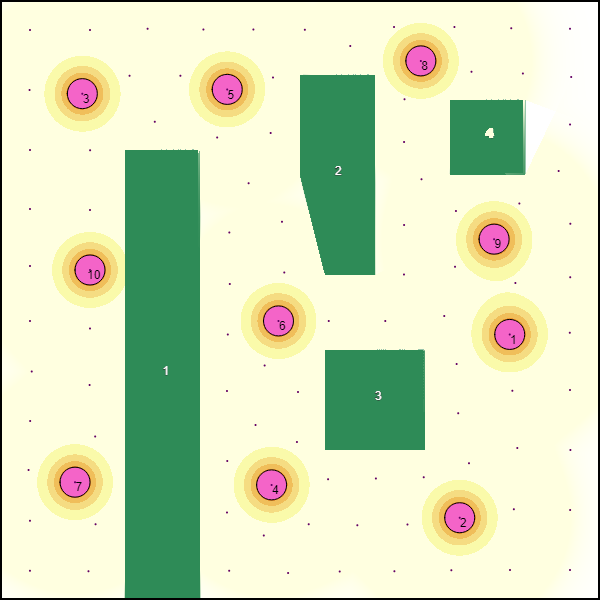}
    \caption{$\lambda = 0.05$}
\end{subfigure}%
\hfill
\begin{subfigure}[t]{0.3\columnwidth}
    \centering
    \includegraphics[width=\columnwidth]{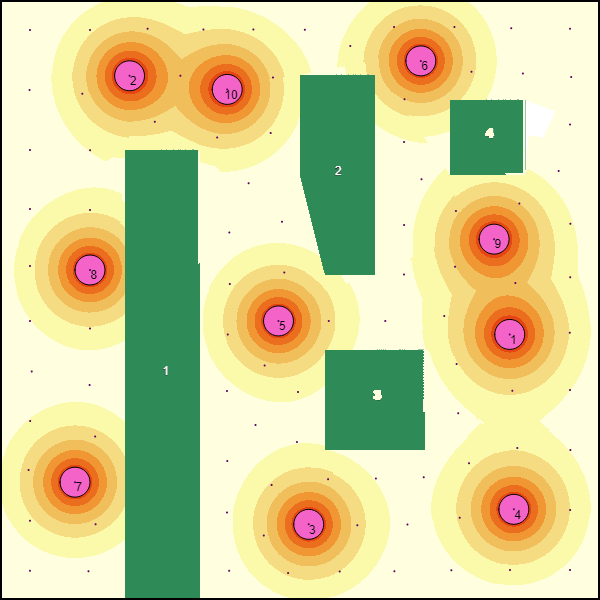}
    \caption{$\lambda = 0.025$}
\end{subfigure}%
\hfill
\begin{subfigure}[t]{0.3\columnwidth}
    \centering
    \includegraphics[width=\columnwidth]{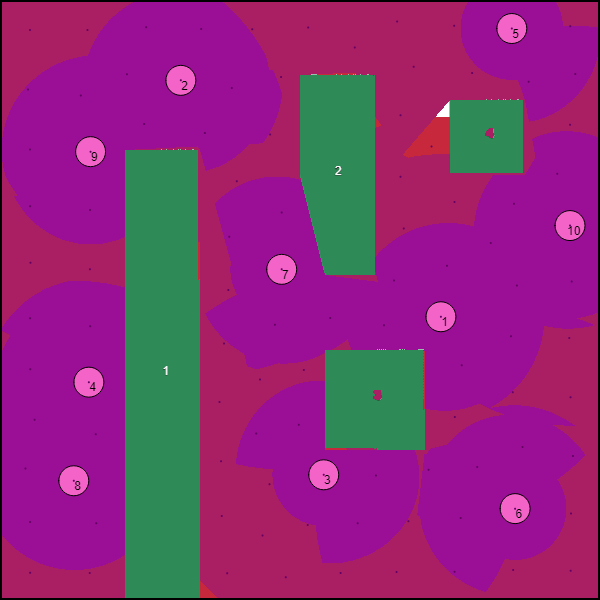}
    \caption{$\lambda = 0.001$}
\end{subfigure}%
\caption{Greedy solutions and coverage level patterns observed under different sensing decay $\lambda$ values considered in Tab. \ref{Tab:GeneralDecay}.}
\label{Fig:GeneralDecay}
\end{figure}

In our numerical experiments, we considered square-shaped mission spaces (each side is of length $600$ units) with three different obstacle arrangements named ``Blank,'' ``Maze'' and ``General'' as can be seen in Figs. \ref{Fig:DetectionFunction}-\ref{Fig:BlankDecay}, \ref{Fig:MazeRange} and \ref{Fig:GeneralDecay}, respectively. Note that, in such figures, obstacles are shown as dark green-colored blocks, candidate agent locations (ground set $X$) are shown as small black dots, and agent locations are shown as numbered pink-colored circles. Note also that light-colored areas indicate low coverage (i.e., low event detection probability) levels, while dark-colored areas indicate the opposite. The event density function was assumed to be uniform, i.e., $R(x) = 1, \forall x\in \Phi$.  

The main attributes and functionalities of the considered class of multi-agent optimal coverage problems (e.g., agent sensing capabilities \eqref{Eq:SensingFunction} and functions like detection \eqref{Eq:DetectionFunction} and coverage \eqref{Eq:CoverageProblem}) as well as the greedy algorithm (Alg. \ref{Alg:GreedyAlgorithm}) and the reviewed performance bounds $\beta_f$ \eqref{Eq:FundamentalPerformanceBound}, $\beta_t$ \eqref{Eq:TotalCurvatureBoundTheory}, $\beta_g$ \eqref{Eq:GreedyCurvatureBoundTheory}, $\beta_e$ \eqref{Eq:ElementalCurvatureBoundTheory}, $\beta_p$ \eqref{Eq:PartialCurvatureBoundTheory} and $\beta_u$ \eqref{Eq:ExtGreedyCurvatureBoundTheory} were all implemented for the considered class of multi-agent optimal coverage problems in an interactive JavaScript-based simulator which is available at \url{https://github.com/shiran27/P2-Submod_Coverage}. This simulator may be used by the reader to reproduce the reported results and also to try different new problem configurations. 

For evaluating the performance bounds $\beta_e$ and $\beta_p$, we used the proposed techniques in Props. \ref{Pr:ElementalCurvatureBound} and \ref{Pr:PartialCurvatureBound}, respectively. Unlike sampling-based techniques used in prior work, these techniques are computationally efficient and provide theoretically valid performance bounds. For evaluating the performance bound $\beta_u$, we used $Q=\bar{Q}$ in \eqref{Eq:ExtGreedyCurvatureMeasure} as in \cite{WelikalaJ02021}. 

\subsection{Impact of the Weight Parameter $\theta$}

First, we demonstrate the effects of the weight parameter $\theta$ (see the detection function in \eqref{Eq:DetectionFunction}) using the Blank mission space with $N=4$ agents. Here, each agent is assumed to have a sensing range $\lambda = 200$ and a sensing decay $\delta = 0.012$ (see the sensing function in \eqref{Eq:SensingFunction}). The observed greedy solution, coverage level pattern, and performance bounds, when this weight parameter is varied from $\theta = 0$ to $\theta=1$, are reported in Fig. \ref{Fig:DetectionFunction}. 

As stated in Rm. \ref{Rm:DetectionFunction}, choosing $\theta=1$ (i.e., promoting joint detection \eqref{Eq:JointDetection}) motivates cooperation while choosing $\theta=0$ (i.e., promoting max detection \eqref{Eq:MaxDetection}) motivates compartmentalization in sensing. This behavior is confirmed by the observations reported in Fig. \ref{Fig:DetectionFunction}. 
In particular, as can be seen in Fig. \ref{Fig:DetectionFunction}(a), when $\theta \simeq 0$, agents are spread out in the mission space - leaving a blind region in the middle but covering a wider area in the mission space. On the other hand, as can be seen in Fig. \ref{Fig:DetectionFunction}(e), when $\theta \simeq 1$, agents are located relatively closer to each other than before - not leaving a blind region in the middle but failing to adequately cover some outer regions of the mission space. 

Besides such implications, as can be seen from the performance bounds reported in Figs. \ref{Fig:DetectionFunction}(a)-(e), the weight parameter $\theta$ also affect the performance bounds. In particular, the performance bound $\beta_u$ (which, in this case, offers the highest performance bound) increases with $\theta$. This behavior implies that, with respect to the used greedy solution approach, the optimal coverage problem defined with a max detection function \eqref{Eq:MaxDetection} is harder to solve than that with a joint detection function \eqref{Eq:JointDetection}. This conclusion is intuitive as max detection functions significantly increase the non-smooth nature of the optimal coverage problems.

Note that, in the numerical examples discussed in the sequel, we predominantly used $\theta=0.5$. However, in a few extreme cases, as can be seen in the top and bottom sub-tables in Tabs. \ref{Tab:BlankRange}-\ref{Tab:GeneralDecay}, we have further investigated the effect of $\theta$ on different performance bounds.

\subsection{Impact of Agent Sensing Capabilities $\delta, \lambda$}

We next demonstrate the impact of the agent sensing capabilities (characterized by their sensing range $\delta$ and sensing decay $\lambda$) on different curvature-based performance bounds established for greedy solutions. For this purpose, we first use the Blank mission space with $N=10$ agents. In one experiment, we kept the agent sensing decay fixed at $\lambda = 0.003$ and varied the sensing range from $\delta = 35$ to $\delta = 800$. The observed performance bounds and a few selected greedy solutions are reported in Tab. \ref{Tab:BlankRange} and the accompanying Fig. \ref{Fig:BlankRange}, respectively. In the second experiment, we kept the agent sensing range fixed at $\delta = 800$ and varied the sensing decay from $\lambda = 0.05$ to $\delta = 0.001$. The observed performance bounds and a few selected greedy solutions are reported in Tab. \ref{Tab:BlankRange} and the accompanying Fig. \ref{Fig:BlankRange}, respectively.

With regard to Tabs. \ref{Tab:BlankRange} and \ref{Tab:BlankDecay}, notice that we have further explored the impact of $\theta$ at the extreme cases, i.e., when $\delta \in \{35,800\}$ and $\lambda \in \{0.05,0.001\}$, respectively. The observations in each case are given in smaller sub-tables located above and below the main table. In each table (and sub-table), the highest performance bound values observed in each row and column have been highlighted for the convenience of the reader. Note also that the tables are arranged in such a way that when going from top to bottom, the sensing capabilities of the agents increase (i.e., $\delta$ increase and $\lambda$ decrease).

Recall that, based on our analysis, the performance bounds $\beta_t$, $\beta_g$, $\beta_p$, and $\beta_u$ should provide significant improvements beyond $\beta_f$ when agent sensing capabilities are low (i.e., when $\delta$ is low and $\lambda$ is high). The results in Tabs. \ref{Tab:BlankRange} and \ref{Tab:BlankDecay} validate this conclusion (e.g., see the respective results for $\delta = 35$ and $\lambda=0.05$). 

On the other hand, our analysis implied that the performance bounds $\beta_e$ and $\beta_u$  should provide significant improvements beyond $\beta_f$ when agent sensing capabilities are high (i.e., when $\delta$ is high and $\lambda$ is low). With regard to $\beta_e$, as mentioned in Rm. \ref{Rm:ElementalCurvatureBoundTriviality}, we also need the additional condition $\theta=1$. Again, the results in Tabs. \ref{Tab:BlankRange} and \ref{Tab:BlankDecay} validate this conclusion (e.g., see the respective results for $\delta=800$ and $\lambda=0.001$, particularly with $\theta = 1$).

Moreover, as expected from our analysis, the observations in Tabs. \ref{Tab:BlankRange} and \ref{Tab:BlankDecay} also confirms that the performance bound $\beta_u$ provide significant improvements beyond $\beta_f$ regardless of the agent sensing capabilities. Of course, there is a small region of moderate agent sensing capabilities for which this improvement is low (yet, still considerable), e.g., see the respective results for $100 \leq \delta \leq 150$ and $0.01 \leq \lambda \leq 0.02$.

Finally, we consider mission spaces with obstacles, in particular, the Maze and General mission spaces, with $N=10$ agents. Parallel to the previous two experiments, in one experiment, we kept the agent sensing decay fixed at $\lambda = 0.012$ and varied the sensing range from $\delta = 50$ to $\delta = 800$. The observed performance bounds and a few selected greedy solutions are reported in Tab. \ref{Tab:MazeRange} and the accompanying Fig. \ref{Fig:MazeRange}, respectively. In the next experiment, we kept the agent sensing range fixed at $\delta = 200$ and varied the sensing decay from $\lambda = 0.05$ to $\delta = 0.001$. The observed performance bounds and a few selected greedy solutions are reported in Tab. \ref{Tab:GeneralDecay} and the accompanying Fig. \ref{Fig:GeneralDecay}, respectively.

Using the results reported in Tabs. \ref{Tab:MazeRange} and \ref{Tab:GeneralDecay}, we can validate all the previous conclusions. The only notable exception here is that the performance bound $\beta_e$ now provides no improvements. This is due to the presence of obstacles that violate the requirement stated in Rm. \ref{Rm:ElementalCurvatureBoundTriviality}.


\section{Conclusion}
\label{Sec:Conclusion}

In this paper, we considered a generalized class of multi-agent optimal coverage problems, combining the max detection and joint detection aspects of coverage applications when defining the coverage objective. We next established submodularity and several other interesting properties of the proposed coverage objective. These properties enabled efficient solving of the considered optimal coverage problem via greedy algorithms with performance-bound guarantees. To obtain further improved performance bounds, we reviewed five curvature measures found in the literature. In particular, we discussed their effectiveness and computational complexity and proposed novel techniques to estimate an efficient candidate for some of such curvature measures. We also implemented the proposed coverage problem in an interactive simulator along with its greedy solutions, and performance bounds. Numerical results given by the simulator validated our earlier conclusions regarding the used detection function and the effectiveness of different performance bounds. Ongoing research explores meaningful ways to combine the strengths of data-driven curvature measures (practicality and effectiveness) with theoretical curvature measures (guarantees).

\bibliographystyle{plainurl}
\bibliography{References} 

\end{document}